\documentclass[11pt]{article}
\usepackage{fullpage}
\usepackage{hdb_macros}
\usepackage{pgfplots}
\usepackage{enumitem}

\pgfplotsset{compat=1.18}

\newcommand{\mat}[1]{\Mat{#1}}
\newcommand{\Mat}[1]{\mathbf{#1}}
\newcommand{\CRP}{\problem{CRP}}
\newcommand{\GapCRP}{\problem{GapCRP}}
\newcommand{\BinGapCRP}{\problem{BinGapCRP}}
\newcommand{\LinDisc}{\problem{LinDisc}}

\newcommand{\basis}{\mat{B}}
\DeclareMathOperator{\lindisc}{lindisc}
\newcommand{\Par}{\mathcal{P}}
\newcommand{\NAEthreeSAT}{\text{\textsc{NAE}-$3$-\textsc{SAT}}}
\newcommand{\NAEEthreeSAT}{\text{\textsc{NAE}-\textsc{E}$3$-\textsc{SAT}}}
\newcommand{\NAEEfourSAT}{\text{\textsc{NAE}-\textsc{E}$4$-\textsc{SAT}}}
\newcommand{\NAEEkSAT}{\text{\textsc{NAE}-\textsc{E}$k$-\textsc{SAT}}}
\newcommand{\AENAEthreeSAT}{\text{$\forall\exists$-\textsc{NAE}-$3$-\textsc{SAT}}}
\newcommand{\AENAEEthreeSAT}{\text{$\forall\exists$-\textsc{NAE}-\textsc{E}$3$-\textsc{SAT}}}

\newcommand{\CSP}{\problem{CSP}}
\renewcommand{\C}{\mathcal{C}}
\newcommand{\val}{\mathrm{val}}
\renewcommand{\E}{\mathbb{E}}
\newcommand{\var}{\mathrm{var}}
\newcommand{\vars}{\mathrm{vars}}
\newcommand{\sgndiff}{\mathrm{sgndiff}}

\DeclarePairedDelimiter{\parens}{(}{)}

\newcommand{\full}[1]{#1}
\newcommand{\fullornot}[2]{#1}

\title{Hardness of the Binary Covering Radius Problem in Large $\ell_p$ Norms}
\fullornot{
\author{Huck Bennett\thanks{University of Colorado Boulder. \email{huck.bennett@colorado.edu}. Supported in part by NSF Award No. 2432132.} \and Peter Ly\thanks{University of Colorado Boulder. \email{peter.ly@colorado.edu}. Supported in part by NSF Award No. 2432132.}}}
{(Anonymous)}
\date{\today}

\begin{document}

\maketitle

\begin{abstract}
We study the hardness of the $\gamma$-approximate decisional Covering Radius Problem on lattices in the $\ell_p$ norm ($\gamma$-$\GapCRP_p$).
Specifically, we prove that there is an explicit function $\gamma(p)$, with $\gamma(p) > 1$ for $p > p_0 \approx 35.31$ and $\lim_{p \to \infty} \gamma(p) = 9/8$, such that for any constant $\eps > 0$, $(\gamma(p) - \eps)$-$\GapCRP_p$ is $\NP$-hard. %
This shows the first hardness of $\GapCRP_p$ for explicit $p < \infty$.
Work of Haviv and Regev (CCC, 2006 and CJTCS, 2012) previously showed $\Pi_2$-hardness of approximation for $\GapCRP_p$ for all sufficiently large (but non-explicit) finite $p$ and for $p = \infty$.

In fact, our hardness results hold for a variant of $\GapCRP$ called the \emph{Binary Covering Radius Problem} ($\BinGapCRP$), which trivially reduces to both $\GapCRP$ and the decisional Linear Discrepancy Problem ($\LinDisc$) in any norm in an approximation-preserving way.
We also show $\Pi_2$-hardness of $(9/8 - \eps)$-$\BinGapCRP$ in the $\ell_{\infty}$ norm for any constant $\eps > 0$. 

Our work extends and heavily uses the work of Manurangsi (IPL, 2021), which showed $\Pi_2$-hardness of $(9/8 - \eps)$-$\LinDisc$ in the $\ell_{\infty}$ norm.
\end{abstract}

\section{Introduction}
\label{sec:intro}

\begin{figure} \label{fig:hardness-plot}
    \centering
        \begin{tikzpicture} %
        \begin{axis}[
            width=8cm,
            height=6cm,
            xmode=log,
            xmin=1,xmax=10000,
            restrict y to domain=0.9:1.15,
            ymin=0.9,ymax=1.15,
            axis x line=middle,
            axis y line=middle,
            axis line style=->,
            xlabel={\(p\)},
            ylabel={\(\gamma\)},
            legend style={at={(0.75,0.5)}, anchor=west}
            ]
            \addplot[no marks,blue,-] expression[domain=1:10000,samples=200,variable=x]
            {((0.9375*(10*(1/2)^x) + 0.0625*((3/2)^x+9*(1/2)^x))/((4/3)^x+3*(2+(4/3)^x)))^(1/x)}; 
            \addplot[no marks,red,-] expression[domain=1:10000,samples=200,variable=x]
            {9/8};
            \fill[blue,thick](axis cs:{(35,1)})circle[radius=1.5pt];
            \legend{\( \gamma(p) \), \( 9/8 \)}
        \end{axis}
    \end{tikzpicture}
    \caption{{\small A linear-log plot of the approximation factor $\gamma = \gamma(p)$ for which we show $\NP$-hardness of $(\gamma - \eps)$-$\BinGapCRP$ in the $\ell_p$ norm (in {\color{blue}blue}; see \cref{thm:intro-np-hardness}). The function $\gamma(p)$, specified in \cref{eq:np-hardness-factor}, is greater than $1$ for all $p > p_0 = 35.310188...$ and satisfies $\lim_{p \to \infty} \gamma(p) = 9/8$. (The {\color{blue}blue dot} is at the point $(p, \gamma) = (p_0, 1)$, where $p_0$ is the unique value of $p$ such that $\gamma(p) = 1$. The horizontal asymptote of $9/8$ is shown in {\color{red}red}.)
    This factor of $9/8$ also appears in our $\Pi_2$-hardness result for $\BinGapCRP_{\infty}$ (\cref{thm:intro-pi2-hardness}) and in the work of Manurangsi~\cite{manurangsi21} on $\LinDisc$, which is the foundation of our work.}}
    \label{fig:hardness-approx-graph}
\end{figure}

The \emph{covering radius} $\mu(\lat)$ of a lattice $\lat$ is the maximum distance of a point in the span of $\lat$ to a vector in $\lat$. It is one of the most fundamental lattice quantities. Accordingly, the (decisional) Covering Radius Problem ((Gap)$\CRP$) is among the most fundamental lattice problems.
Indeed, it is connected to the main applications of lattices in computer science. Approximate $\GapCRP$ is one of the worst-case problems from which there is a reduction to the average-case problems used in lattice-based cryptography~\cite{journals/siamcomp/Micciancio04}, and it is related to fast algorithms for integer programming (see, e.g.,~\cite{conf/stoc/Dadush19}).

Yet, the hardness of $\GapCRP$ is perhaps the least well understood among major lattice problems.
In particular, it is not even known that \emph{exact} $\GapCRP$ in the $\ell_2$ norm is $\NP$-hard, even though it is plausibly $\NP$-hard to approximate within any constant $\gamma \geq 1$ and $\Pi_2$-hard to approximate to within a small constant $\gamma > 1$.
This stands in contrast to, e.g., the decisional Shortest Vector Problem ($\gamma$-$\GapSVP$), which a long line of elegant work~\cite{DBLP:conf/stoc/Ajtai98,journals/siamcomp/Micciancio00,journals/jcss/Khot06,journals/jacm/Khot05,conf/stoc/RegevR06,journals/toc/HavivR12,journals/toc/Micciancio12} showed to be $\NP$-hard in any $\ell_p$ norm for any constant $\gamma \geq 1$.\footnote{Formally, this $\NP$-hardness is only known under \emph{randomized reductions}, although several works---some of which are very recent~\cite{hair-sahai-detsvp25,hechtsafra2025deterministichsvp}!---have derandomized the reductions in some cases.}

Seminal work of Haviv and Regev~\cite{journals/cjtcs/HavivR12} was the first---and prior to this work, only---to study the hardness of $\GapCRP$. 
(Other work has studied the complexity of $\GapCRP$ more generally, including limitations on its hardness~\cite{guruswami-micciancio-regev,Magazinov-sqrt3CRP}. See \cref{sec:related-work}.)
Let $\gamma$-$\GapCRP_p$ be $\gamma$-approximate $\GapCRP$ in the $\ell_p$ norm.
Haviv and Regev~\cite{journals/cjtcs/HavivR12} showed that for any constant $\eps > 0$, $(3/2 - \eps)$-$\GapCRP_{\infty}$ is $\Pi_2$-hard, and that for all sufficiently large $p$ there exists $\gamma = \gamma(p) > 1$ such that $\gamma$-$\GapCRP_p$ is $\Pi_2$-hard.
However, the cutoff for ``sufficiently large'' in their work is non-explicit and presumably quite high, and so in some sense not much more is known about the hardness of $\GapCRP_p$ for any explicit, finite $p$ than for $p = 2$.

\subsection{Our Results}
\label{sec:our-work}

In this work, we address the problem above by showing $\NP$-hardness of $(\gamma - \eps)$-$\GapCRP_p$ for all $p$ greater than an \emph{explicit} value $p_0 \approx 35.31$ with an explicit approximation factor $\gamma = \gamma(p) > 1$. This is the first new hardness result for $\GapCRP$ since~\cite{journals/cjtcs/HavivR12}, which was first published $20$ years ago.

In fact, we show hardness for a new promise variant of $\gamma$-$\GapCRP_p$ called $\gamma$-$\BinGapCRP_p$ where in the YES case every vector in the fundamental parallelepiped $\Par(\basis) := \basis \cdot [0, 1]^n$ of the input basis $\basis \in \R^{m \times n}$ is close in $\ell_p$ distance to a lattice vector $\basis \vec{x}$ with a \emph{binary} coefficient vector $\vec{x} \in \bit^n$. (In the NO case, there exists a vector in $\Par(\basis)$ that is far from \emph{every} lattice vector $\basis \vec{x}$ for $\vec{x} \in \Z^n$, as with standard $\gamma$-$\GapCRP_p$.)

$\BinGapCRP$ is reminiscent of the ``binary coefficient vector'' variant of the Closest Vector Problem often used in hardness reductions (see, e.g.,~\cite[Definition 7]{GupteV-fine-grained-lin-reg} and~\cite[Definition 3.1]{journals/sigact/Bennett23}), and may be of independent interest.
It has the YES instances of the Linear Discrepancy Problem ($\LinDisc$) and the NO instances of $\GapCRP$. It therefore trivially reduces to both problems in any norm in an approximation-preserving way (see~\cref{sec:prelims-comp-covrad-lindisc} for definitions of these problems).
In particular, our result not only implies new hardness for $\GapCRP_p$ with $p < \infty$, but also for $\LinDisc$ in the $\ell_p$ norm for $p < \infty$ ($\LinDisc$ in finite $\ell_p$ norms does not seem to have been studied before).

Our main result is the following theorem. See also \cref{fig:hardness-approx-graph}.
\begin{theorem}[$\NP$-hardness of approximation for $\BinGapCRP_p$] \label{thm:intro-np-hardness}
Let 
\begin{equation} \label{eq:np-hardness-factor}
\gamma(p) := \Big(\frac{9^p + 159 \cdot 3^p}{8^{p+2} + 96 \cdot 6^p}\Big)^{1/p} \ \text{,}
\end{equation}
and let $p_0 = 35.310188\ldots$ be the unique value of $p \geq 1$ such that $\gamma(p) = 1$.
Then for all \( p > p_0 \) and any constant $\eps > 0$, $(\gamma(p) - \eps)$-$\BinGapCRP_p$ is $\NP$-hard.
\end{theorem}

We note that the expression $\gamma(p)$ in \cref{eq:np-hardness-factor} satisfies $\lim_{p \to \infty} \gamma(p) = 9/8$. 
Our second result shows how to achieve $\Pi_2$-hardness (rather than just $\NP$-hardness) of $(9/8 - \eps)$-$\BinGapCRP$ in the $\ell_{\infty}$ norm.

\begin{theorem}[$\Pi_2$-hardness of approximation for $\BinGapCRP_{\infty}$] \label{thm:intro-pi2-hardness}
For any constant $\eps > 0$,
$(9/8 - \eps)$-$\BinGapCRP_{\infty}$ is $\Pi_2$-hard.
\end{theorem}

\cref{thm:intro-pi2-hardness} is a strengthening of the main result in \cite{manurangsi21}---it shows hardness not just of $(9/8 - \eps)$-$\LinDisc_{\infty}$ but of $(9/8 - \eps)$-$\BinGapCRP_{\infty}$, where $\LinDisc_{\infty}$ is the Linear Discrepancy Problem in the $\ell_{\infty}$ norm.
\cref{thm:intro-pi2-hardness} also gives an alternative proof of $\Pi_2$-hardness of approximation for $\gamma$-$\GapCRP_{\infty}$, albeit with a smaller approximation factor than the one achieved by~\cite{journals/cjtcs/HavivR12}.

\subsection{Related Work}
\label{sec:related-work}

Early work of McLoughlin~\cite{mcloughlin} showed that the Covering Radius Problem on error-correcting codes is \( \Pi_2 \)-hard by a reduction from a \( \Pi_2 \)-hard variant of the $3$-Dimensional Matching Problem, which is in turn shown to be \( \Pi_2 \)-hard by a reduction from $\Pi_2$-$3$-SAT.%
\footnote{Here and throughout this section we refer to linear error-correcting codes simply as ``codes.''}
Twenty years later, work of Guruswami, Micciancio, and Regev~\cite{guruswami-micciancio-regev} studied the complexity of the Covering Radius Problem both on codes and lattices.
It showed a number of results, including that the Covering Radius Problem on codes is $\NP$-hard to approximate to within any constant and $\Pi_2$-hard to approximate to within some constant. 
(It did not show hardness results for $\GapCRP$ on lattices.)

The work~\cite{guruswami-micciancio-regev} also showed that $2$-$\GapCRP_p$ for any $p$ is in $\AM$ via an elegant protocol. 
The fact that $2$-$\GapCRP_p$ is in $\AM$ provides evidence that $\gamma$-$\GapCRP_p$ for $\gamma \geq 2$ is unlikely to be $\Pi_2$-hard, since if it were the polynomial hierarchy would collapse.
However, being in $\AM$ says nothing about the potential $\NP$-hardness of these problems. In particular it could very well be the case that $\gamma$-$\GapCRP$ is $\NP$-hard for every constant $\gamma \geq 1$ as is the case for the Covering Radius Problem on codes.

The factor of $\gamma = 2$ in the~\cite{guruswami-micciancio-regev} $\AM$ protocol follows generically for $\GapCRP$ in any norm from the triangle inequality, but it is not always tight.
Follow-up work by Haviv, Lyubashevsky, and Regev~\cite{journals/dcg/HavivLR09} studied the distribution of the (normalized) $\ell_p$ distance of a random point in a fundamental domain of a lattice to the lattice.
This distribution is closely related to the~\cite{guruswami-micciancio-regev} $\AM$ protocol, and~\cite{journals/dcg/HavivLR09} showed that $\gamma$-$\GapCRP_2$ for $\gamma > \sqrt{3} \approx 1.732$ is in $\AM$ under a certain conjecture about this distribution in the $\ell_2$ norm.
Their conjecture was subsequently proved by Magazinov~\cite{Magazinov-sqrt3CRP}. %

Understanding the worst-case hardness of $\gamma$-$\GapCRP$ is also motivated by its relevance to cryptography. Indeed,~\cite{journals/siamcomp/Micciancio04} showed that $\gamma$-$\GapCRP$ with sufficiently large polynomial $\gamma = \gamma(n)$ can be taken as the starting worst-case problem for a worst-case to average-case reduction to the problems used as the foundation of lattice-based cryptography.

\paragraph{The Haviv-Regev hardness reduction.}
The work of Haviv and Regev~\cite{journals/cjtcs/HavivR12} is the only prior work that studied the worst-case hardness of $\GapCRP$ on lattices. As noted above, they showed that for every constant $\eps > 0$, $(3/2 - \eps)$-$\GapCRP_{\infty}$ is $\Pi_2$-hard and that for all sufficiently large $p$ there exists $\gamma = \gamma(p)$ such that $\gamma$-$\GapCRP_p$ is $\Pi_2$-hard to approximate. However, the cutoff for ``sufficiently large'' and the approximation factors $\gamma(p)$ for $p < \infty$ are non-explicit.

The reduction in~\cite{journals/cjtcs/HavivR12} is from the Group Coloring Problem on graphs with respect to the group $\Z_3$, a problem that is $\Pi_2$-hard.\footnote{Their reduction actually works with $\Z_q$ for any $q \geq 3$, but they achieve the best hardness of approximation with $q = 3$.}
The problem is defined as follows. The input is a directed graph $G = (V, E)$. The goal is to decide whether for every assignment $\phi : E \to \Z_3$ of labels to edges, there exists a coloring $c : V \to \Z_3$ of vertices such that for every edge $e = (u, v) \in E$, $c(u) - c(v) \neq \phi(e)$. The reduction to $\GapCRP$ is as follows. Let $A = A(G) \in \Z^{m \times n}$ for $m := \card{E}$, $n := \card{V}$ be the incidence matrix of $G$. That is, $A$ is the matrix with two non-zero entries per row defined as
\[
A_{i,j} := \begin{cases}
-1 & \text{if edge $e_i = (v_j, u)$ for some $u \in V$,} \\
1 & \text{if edge $e_i = (u, v_j)$ for some $u \in V$,} \\
0 & \text{otherwise.}
\end{cases}
\]
The reduction then outputs (a basis of) the lattice $\lat := A \cdot \Z^n + 3 \Z^m$.
One may also view $A$ as the generator matrix of a ternary code $\C \subseteq \F_3^m$, and $\lat$ as the Construction-A lattice $\C + 3\Z^m$.

To show $\NP$-hardness of $\gamma$-$\GapCRP_p$ in some $\ell_p$ norm (for large $\gamma$), one might try to reduce from variants of the $3$-Coloring Problem, which is the $\NP$ analog of the $\Pi_2$-complete $\Z_3$-Group Coloring Problem.
Indeed, strong hardness of approximation results are known for Almost-$3$-Coloring---the optimization variant of $3$-Coloring in which the goal is to minimize the number of monochromatic edges---assuming variants of the Unique Games Conjecture~\cite{journals/siamcomp/DinurMR09}.
However, although an analog of the soundness argument in~\cite{journals/cjtcs/HavivR12} holds for (Almost-)$3$-Coloring, it is unclear how to show an analog of the completeness argument.

\paragraph{Combinatorial and linear discrepancy.}
The (combinatorial) \emph{discrepancy} of a set system $S_1, \ldots, S_m \subseteq \set{x_1, \ldots, x_n}$ is defined as $\min_{x_1, \ldots, x_n \in \pmo} \max_{i \in [m]} \abs{\sum_{j \in S_i} x_j}$, and the goal of the \emph{combinatorial discrepancy problem} is, given the set system as input, to find $x_1, \ldots, x_n \in \pmo$ that achieve this minimum.
Equivalently, if $\mat{A} \in \bit^{m \times n}$ is the incidence matrix of the input set system, the goal is to find $\vec{x} \in \pmo^n$ that minimizes $\norm{\mat{A} \vec{x}}_{\infty}$.
Combinatorial discrepancy and other variants of discrepancy are used widely in computer science and math (see, e.g.,~\cite{panorama-discrepancy-book}).

\emph{Linear discrepancy} was originally defined and related to other notions of discrepancy in work of Lov\'{a}sz, Spencer, and Vestergombi~\cite{LovaszSV/discrepancy86}, and it appears as a key ingredient in certain approximation algorithms (see, e.g.,~\cite{journals/talg/EisenbrandPR13,journals/siamcomp/Rothvoss16,conf/soda/HobergR17}).
Formally, we define the linear discrepancy $\lindisc(\basis)$ of a matrix $\basis$ in the $\ell_p$ norm as\footnote{``Standard'' linear discrepancy corresponds to $\lindisc_{\infty}$, and, to the best of our knowledge, our work is the first to extend the definition to finite $\ell_p$ norms. This extension is useful for our work.}
\begin{equation} \label{eq:lin-disc-intro}
\lindisc_p(\basis) := \max_{\vec{w} \in [0, 1]^n} \min_{\vec{x} \in \bit^n} \norm{\basis (\vec{w} - \vec{x})}_p \ \text{.}
\end{equation}
This definition is syntactically very similar to the definition of the covering radius in the $\ell_p$ norm, $\mu_p$.
The difference is that in the definition of $\mu_p$, the minimum is taken over $\vec{x} \in \Z^n$ rather than $\vec{x} \in \bit^n$. See \cref{sec:lattices-lin-disc,sec:prelims-comp-covrad-lindisc}.

Combinatorial discrepancy is captured, up to translation and scaling, by fixing $\vec{w} := \half \cdot \vec{1}$ in \cref{eq:lin-disc-intro} instead of taking the max over all $\vec{w} \in [0, 1]^n$ there.
Linear discrepancy was first investigated from a complexity perspective by Li and Nikolov~\cite{Li-Nikolov}, who provided both algorithms and hardness results. 
In particular, they showed that the exact Linear Discrepancy Problem in the $\ell_{\infty}$ norm, $\LinDisc_{\infty}$, is $\NP$-hard and that it is contained in $\Pi_2$.
Their hardness result was subsequently improved by Manurangsi~\cite{manurangsi21}, who showed that $\gamma$-$\LinDisc_{\infty}$ is $\Pi_2$-hard for $\gamma := 9/8 - \eps$ for any constant \( \eps > 0 \). Our work makes extensive use of~\cite{manurangsi21}.

\subsection{Technical Overview}
\label{sec:techniques}

At a high level, our hardness reduction for $\BinGapCRP_p$ follows the hardness reduction for $\LinDisc_{\infty}$ in~\cite{manurangsi21}.
The reduction in~\cite{manurangsi21} reduces from $\NAEthreeSAT$, the constraint satisfaction problem in which each constraint is a function of three literals (Boolean variables $v_j$ or their negations $\lnot v_j$), and a constraint is satisfied by an assignment to its variables if and only if the assignment neither satisfies none of its literals nor all of its literals.

Specifically, on input a $\NAEthreeSAT$ formula $\phi(v_1, \ldots, v_n)$ with $n$ Boolean variables and with $m$ constraints, the $\NP$-hardness reduction in~\cite{manurangsi21} first computes the constraint-variable incidence matrix $\mat{B} = \mat{B}(\phi) \in \Z^{m \times n}$ with entries\footnote{In fact,~\cite{manurangsi21} reduces from the $\NAEthreeSAT$ (rather than $\NAEEthreeSAT$) problem, in which constraints may have two literals rather than three. The reduction simply doubles an arbitrary literal in any constraint with two literals, which leads to a matrix $\mat{B}$ potentially containing entries from $\set{-2, 2}$ as well as $\set{-1, 0, 1}$. In this work, we reduce from $\NAEEthreeSAT$.}
\[
    B_{i,j} :=
    \begin{cases}
        1 & \text{if constraint \( i \) contains \( v_j \) ,} \\
        -1 & \text{if constraint \( i \) contains \( \lnot v_j \) ,} \\
        0 & \text{otherwise .}
    \end{cases}
\]
It then outputs a $\LinDisc_{\infty}$ instance with matrix
\begin{equation} \label{eq:matA-intro}
    \mat{A} := \begin{pmatrix}
        \frac{1}{3}\mat{B} & \frac{1}{3}\mat{B} & \frac{1}{3}\mat{B} \\
        & \mat{G} \otimes \mat{I}_n
    \end{pmatrix}
    \in \Q^{(m+3n) \times 3n} \ \text{,}
\end{equation}
where $\mat{G}$ is the $3 \times 3$ gadget matrix
\[
\mat{G} := \begin{pmatrix}
1 & 1 & -1 \\
1 & -1 & 1 \\
-1 & 1 & 1
\end{pmatrix} \ \text{,}
\]
$\otimes$ is the Kronecker product and $\mat{I}_n$ is the $n \times n$ identity matrix.

Both the completeness and soundness analysis of the reduction in~\cite{manurangsi21} crucially use properties of $\mat{G}$.
For example,~\cite[Lemma 3]{manurangsi21} argues that $\mat{G}$ has low discrepancy by upper bounding the distance between $\mat{G} \vec{u}$ and $\mat{G} \cdot \bit^3$ for each $\vec{u} \in [0, 1]^3$ via a case analysis on the $\ell_1$ norm of $\vec{u}$.
The soundness analysis shows that, if there exists a close vector in the set $\mat{A} \cdot \bit^{3n}$ to $\mat{A} \cdot (\frac{1}{2} \vec{1})$, then $\phi$ must be satisfiable.
\cite{manurangsi21} also shows $\Pi_2$-hardness of $(9/8 - \eps)$-$\LinDisc_{\infty}$ by outputting and analyzing a matrix $\mat{A}' = \mat{A}'(\phi)$ defined in terms of $\phi$ that is related to the matrix $\mat{A}$ defined in \cref{eq:matA-intro} but more complicated. We use this same matrix $\mat{A}'$ (defined in \cref{eqn:def-matAprime}) to prove \cref{thm:intro-pi2-hardness}.

At a technical level, our $\NP$-hardness result adapts the machinery in~\cite{manurangsi21} from $\ell_{\infty}$ to general $\ell_p$ norms, and shows that, if the $\NAEEthreeSAT$ formula $\phi$ is unsatisfiable, then not only is $\lindisc(\mat{A})$ large, in fact $\mu(\lat(\mat{A}))$ is large.
That is, we show that if $\phi$ is unsatisfiable then there exists $\vec{w} \in [0, 1]^{3n}$ (in fact, $\vec{w} := \frac{1}{2} \cdot \vec{1}$) such that $\norm{\mat{A}(\vec{w} - \vec{x})}_p$ is large not only for all $\vec{x} \in \bit^{3n}$ but for all $\vec{x} \in \Z^{3n}$.
Doing these things requires making a number of modifications to the hardness reduction and analysis used by \cite{manurangsi21}, such as
\begin{enumerate}

\item \label{item:nae3sat-hardness} Using \emph{hardness of approximation} for $\NAEEthreeSAT$ with soundness $15/16 + \eps$ and perfect completeness; see \cref{thm:nae-e3-sat-hard}. Although this result is likely folklore, we did not find a suitable reference for it and give a proof in \cref{sec:nae-e3-sat-hardness}.

\item Taking the Kronecker product of $G$ with the ``$\ell_p$ norm degree'' diagonal matrix
\[
    \mat{D}_p := \mat{D}_p(\phi) = 
    \diag(\deg(v_1)^{1/p}, \ldots, \deg(v_n)^{1/p}) \in \R^{n \times n}
\]
in the definition of $\mat{A}$ instead of just taking the Kronecker product with $I_n$.
Here $\deg(v_i)$ counts the number of occurrences of the variable $v_i$ and its negation $\lnot v_i$ in the formula $\phi$.

\item \label{item:G-analysis-gen} 
Extending the analysis from~\cite{manurangsi21} for the gadget matrix $\mat{G}$. For example, we prove an $\ell_p$ analog of the fact that $\lindisc(\mat{G})_{\infty}$ is small (and in fact, that something stronger is true) in \cref{lem:gadget-properties}, and we observe that $\mat{G} \cdot \vec{0} = \vec{0}$ and $\mat{G} \cdot \vec{1} = \vec{1}$ are the only closest vectors in $\lat(\mat{G}) = \mat{G} \cdot \Z^3$ to $\mat{G} \cdot (\half \vec{1})$ (and not just the only closest vectors in $\mat{G} \cdot \bit^3$ to $\mat{G} \cdot (\half \vec{1})$) in \cref{lem:gadget-property-soundness}.
\end{enumerate}

For \cref{item:nae3sat-hardness}, we note that the soundness-to-completeness ratio of $15/16$ in \cref{thm:nae-e3-sat-hard} is \emph{not} tight. In fact, recent work of Brakensiek, Huang, Potechin, and Zwick~\cite{brakensiek-huang-potechin-zwick} shows tight upper and lower bounds for $\NAEEthreeSAT$ with a lower (i.e., better) soundness-to-completeness ratio of roughly $0.9089 < 15/16 = 0.9375$.
However, \cref{thm:nae-e3-sat-hard} shows hardness of $(\delta, \eps)$-$\NAEEthreeSAT$ \emph{with perfect completeness} (i.e., with $\eps = 1$), whereas the result of~\cite{brakensiek-huang-potechin-zwick} does not.
This is important since our main hardness reduction for $\BinGapCRP_p$ requires reducing from $(\delta, \eps)$-$\NAEEthreeSAT$ with nearly perfect completeness.

\subsection{Open Questions}
\label{sec:open-questions}

We again emphasize the central open question that we consider in this work: For which values of $\gamma = \gamma(p)$ is $\gamma$-$\GapCRP$ in the $\ell_p$ norm $\NP$-hard, and for which is it $\Pi_2$-hard? In particular, we reiterate that the problem of showing $\NP$-hardness of exact $\GapCRP$ in the $\ell_2$ norm (the most important special case) remains open.
We again note that the $\AM$ protocol of~\cite{guruswami-micciancio-regev} for approximate $\GapCRP$ (conditionally) rules out $\Pi_2$-hardness for $\gamma$-$\GapCRP_p$ for $\gamma \geq 2$, but says nothing about $\NP$-hardness of $\gamma$-$\GapCRP_p$.
Manurangsi~\cite{manurangsi21} notes that no such analogous $\AM$ protocol is known for $\LinDisc$, and asks if one exists.

\full{\subsection{Acknowledgements} We thank Ishay Haviv, Pasin Manurangsi, Daniele Micciancio, and Noah Stephens-Davidowitz for helpful comments and references.}

\section{Preliminaries}
\label{sec:prelims}

We will use boldface, lowercase letters to denote vectors. We will sometimes abuse notation when defining column vectors in terms of other column vectors by writing things like $\vec{x} := (\vec{y}, \vec{z})$ to mean $\vec{x} := (\vec{y}^T, \vec{z}^T)^T$. We will use $\vec{0}$ and $\vec{1}$ to denote the all-$0$s and all-$1$s vectors, respectively.

We define the predicate $\sgndiff(\vec{x}, \vec{y})$ on vectors $\vec{x}, \vec{y} \in \R^n$ to be true if and only if the Hadamard (entry-wise) product $\vec{x} \odot \vec{y} \in \R^n$ has at least one non-positive coordinate and at least one non-negative coordinate.
We will use the simple fact that if
$\sgndiff(\vec{x}, \vec{y})$ holds and if $\max_i \abs{x_i y_i} \leq r$ for some $r \geq 0$, then $\abs{\iprod{\vec{x}, \vec{y}}} = \abs{\sum_{i=1}^n x_i y_i} \leq r(n-1)$.

\subsection{Covering Radius and Linear Discrepancy}
\label{sec:lattices-lin-disc}

We call a matrix $\basis = (\vec{b}_1, \ldots, \vec{b}_n) \in \R^{m \times n}$ with full column rank a \emph{basis}. In particular, $m \geq n$ for a basis.
The lattice $\lat(\basis)$ generated by $\basis$ is defined as
\[
\lat(\basis) := \big\{\sum_{i=1}^n a_i \vec{b}_i : a_1, \ldots, a_n \in \Z\big\} \ \text{.}
\]
For $p \geq 1$, the $\ell_p$ \emph{covering radius} of a lattice $\lat \subset \R^m$ is defined as
\[
\mu_p(\lat) := \max_{\vec{t} \in \lspan(\lat)} \dist_p(\vec{t}, \lat) \ \text{,}
\]
where for a discrete set $S$, we define $\dist_p(\vec{t}, S) := \min_{\vec{y} \in S} \norm{\vec{t} - \vec{y}}_p$.
Because $\dist_p(\vec{t} + \vec{y}, \lat) = \dist_p(\vec{t}, \lat)$ for any $\vec{t} \in \R^m$ and $\vec{y} \in \lat$, we have that if $\lat = \lat(\basis)$ for a basis $\basis \in \R^{m \times n}$, 
\begin{equation} \label{eq:covrad-prelims}
\mu_p(\lat) 
= \max_{\vec{w} \in \R^n} \min_{\vec{x} \in \Z^n} \norm{\basis (\vec{w} - \vec{x})}_p
= \max_{\vec{w} \in [0, 1]^n} \min_{\vec{x} \in \Z^n} \norm{\basis (\vec{w} - \vec{x})}_p \ \text{.}
\end{equation}
That is, because the distance of a vector $\basis \vec{w}$ to $\lat(\basis)$ is invariant under shifts by lattice vectors (i.e., vectors $\basis \vec{x}$ for $\vec{x} \in \Z^n$), it suffices to consider the maximum distance of vectors $\basis \vec{w}$ in the fundamental parallelepiped $\Par(\basis) := \basis \cdot [0, 1]^n$ to the lattice.

We will also work with \emph{linear discrepancy}, a quantity closely related to the covering radius that was first defined in~\cite{LovaszSV/discrepancy86}. For $p \geq 1$, the linear discrepancy $\lindisc(\basis)$ of a matrix $\basis$ is defined as\footnote{The definition of $\lindisc(\basis)$ does not require $\basis$ to be a basis. I.e., $\basis$ is allowed to have linearly dependent columns.}
\begin{equation} \label{eq:lindisc-prelims}
\lindisc_p(\basis) := \max_{\vec{w} \in [0, 1]^n} \min_{\vec{x} \in \bit^n} \norm{\basis (\vec{w} - \vec{x})}_p \ \text{.}
\end{equation}

\paragraph{Covering radius versus linear discrepancy.}
We note that the equivalent formulation of $\mu_p(\lat)$ on the right-hand side of \cref{eq:covrad-prelims} and the definition of $\lindisc(\basis)$ are syntactically very similar. The only difference is that the minimum in \cref{eq:covrad-prelims} is taken over $\vec{x} \in \Z^n$ (which corresponds to the entire lattice $\lat(\basis) = \set{\basis \vec{x} : \vec{x} \in \Z^n}$) whereas in \cref{eq:lindisc-prelims} the minimum is only taken over $\vec{x} \in \bit^n$ (which corresponds to \emph{binary} linear combinations $\basis \vec{x}$ of basis vectors).

Despite the similarity of the definitions of $\mu_p$ and $\lindisc_p$, we emphasize that they capture different quantities.
In particular, the closest vector in the lattice $\lat(\basis) := \basis \cdot \Z^n$ to a vector $\basis \cdot \vec{w}$ for $\vec{w} \in [0, 1]^n$ need not be a vector in $\basis \cdot \bit^n$. 
 We give a simple example of this for the lattice $\Z^2$ when the input basis $\basis$ consists of rather long vectors.
\begin{example}
    Consider the basis
    \[
    \basis := \begin{pmatrix}
    3 & 4 \\
    1 & 1
    \end{pmatrix} \ \text{,}
    \]
    and note that $\lat(\basis) = \Z^2$ since $\det(\basis) = -1$ (i.e., $\basis$ is unimodular).
    Let $\vec{w} := (1, 1/2)^T$, and note that $\basis \vec{w} = (5, 3/2)^T$. Then
    \[
    \dist_p(\basis \vec{w}, \lat(\basis)) = \dist_p(\basis \vec{w}, \Z^2) = \norm{\basis((3, -1)^T - \vec{w})}_p = 1/2 \ \text{,}
    \]
    but one can check that $\norm{\basis (\vec{x} - \vec{w})}_p > 1/2$ for all $\vec{x} \in \bit^2$.
\end{example}

\subsection{Covering Radius and Linear Discrepancy as Computational Problems}
\label{sec:prelims-comp-covrad-lindisc}

We next define the approximate decisional Covering Radius Problem in the $\ell_p$ norm, $\gamma$-$\GapCRP$, and the closely related approximate Linear Discrepancy Problem in the $\ell_p$ norm, $\gamma$-$\LinDisc_p$. %

\begin{definition} \label{def:gapcrp}
Let $p \in [1, \infty]$ and let $\gamma = \gamma(n) \geq 1$.
    The \emph{$\gamma$-approximate decisional Covering Radius Problem in the $\ell_p$ norm} ($\gamma$-$\GapCRP_p$), is the problem where, on input a basis $\basis \in \Q^{m \times n}$ and a value $r > 0$, the goal is to distinguish between the following two cases.

    \begin{itemize}
    \item (YES instance.)  $\mu_p(\lat(\basis)) \leq r$.
    
    \item (NO instance.)  $\mu_p(\lat(\basis)) > \gamma r$.
    \end{itemize}
\end{definition}

\begin{definition} \label{def:lindisc}
Let $p \in [1, \infty]$ and let $\gamma = \gamma(n) \geq 1$.
    The \emph{$\gamma$-approximate decisional Linear Disrepancy Problem in the $\ell_p$ norm} ($\gamma$-$\LinDisc_p$) is the problem where, on input a basis $\basis \in \Q^{m \times n}$ and a value $r > 0$, the goal is to distinguish between the following two cases.

    \begin{itemize}
    \item (YES instance.) $\lindisc_p(\basis) \leq r$.
    
    \item (NO instance.) $\lindisc_p(\basis) > \gamma r$.
    \end{itemize}
\end{definition}

Finally, we define the binary decisional Covering Radius Problem $\BinGapCRP$, which has the YES instances of $\LinDisc$ and the NO instances of $\GapCRP$.

\begin{definition} \label{def:bin-gap-crp}
    Let $p \in [1, \infty]$ and let $\gamma = \gamma(n) \geq 1$.
    The \emph{$\gamma$-approximate decisional Binary Covering Radius Problem in the $\ell_p$ norm} ($\gamma$-$\BinGapCRP_p$) is the problem where, on input a basis $\basis \in \Q^{m \times n}$ and a value $r > 0$, the goal is to distinguish between the following two cases.
    \begin{itemize}
    \item (YES instance.) $\lindisc_p(\basis) \leq r$.
    
    \item (NO instance.)  $\mu_p(\lat(\basis)) > \gamma r$.

    \end{itemize}
\end{definition}

We note that $\mu_p(\lat(\basis)) \leq \lindisc_p(\basis)$ for any $p \in [1, \infty]$ and any basis $\basis$.
Therefore, for any $p \in [1, \infty]$ and $\gamma \geq 1$, YES (respectively, NO) instances of $\gamma$-$\BinGapCRP_p$ are YES (respectively, NO) instances of both $\gamma$-$\GapCRP_p$ and $\gamma$-$\LinDisc_p$. We therefore immediately have the following observation.

\begin{fact} \label{fct:bingapcrp-to-crp-lindisc}
For any $\gamma \geq 1$ and $p \in [1, \infty]$, the identity mapping $(\basis, r) \mapsto (\basis, r)$ is a (polynomial-time) reduction from $\gamma$-$\BinGapCRP_p$ to $\gamma$-$\GapCRP_p$ and to $\gamma$-$\LinDisc_p$.
\end{fact}

\subsection{Constraint Satisfaction Problems}
\label{sec:prelims-csps}

A (Boolean) \emph{constraint satisfaction problem} (CSP) is defined by a set of Boolean variables \( X = \set{x_1, \dots, x_n} \), and a set of functions called \emph{constraints} \( \mathcal{C} = \set{C_1, \dots, C_m} \) with \( C_i \colon \bit^{\abs{C_i}} \to \bit \). An \emph{assignment} of the variables is a map from \( X \to \bit \). A constraint \( C_i \) is \emph{satisfied} by an assignment if the output of the constraint on the assignment of the variables is \( 1 \).
Let \( \phi \) be a CSP and let \( \psi \) be an assignment of the variables in \( \phi \). We define \( \val_{\psi}(\phi) \) to be the fraction of constraints in \( \phi \) satisfied by an assignment \( \psi \), and define the \emph{value} of $\phi$ to be \( \val(\phi) := \max_{\psi} \val_{\psi}(\phi) \).
We then define the computational problem $(\delta, \eps)$-$\CSP$ as follows.

\begin{definition}[Constraint Satisfaction Problems (CSP)]
    For $\delta, \eps \in \R$ be such that $0 < \delta \leq \eps \leq 1$,
    $(\delta, \eps)$-$\CSP$ is the decision problem defined as follows. 
    On input a collection of constraints $\C = \set{C_1, \ldots, C_m}$, the goal is to decide between the follows two cases.
    \begin{itemize}
    \item (YES instance.) $\val(\phi) \geq \eps$.
    \item (NO instance.) $\val(\phi) < \delta$.
    \end{itemize}
\end{definition}

We now introduce $\NAEEkSAT$, a $k$-ary CSP that will be important for our hardness reductions. 

\begin{definition}[\( (\delta, \eps)\)-$\NAEEkSAT$ problem]
    A $\NAEEkSAT$ formula $\phi$ consists of $n$ Boolean variables $X = \set{x_1, \ldots, x_n}$ and $m$ constraints $C_1, \ldots, C_m$, where each constraint $C_i$ is a function of exactly $k$ distinct literals (Boolean variables $x_j$ or their negations $\lnot x_j$). A constraint $C_i$ is satisfied by an assignment if and only if not all $k$ of its literals evaluate either to $1$ or to $0$.

    Let $\delta, \eps \in \R$ be such that $0 < \delta \leq \eps \leq 1$. 
    The goal of the $(\delta, \eps)$-$\NAEEkSAT$ problem is, given a $\NAEEkSAT$ formula as input, to decide between the following two cases.
    
    \begin{itemize}
        \item (YES instance.) $\val(\phi) \geq \eps$.
        \item (NO instance.) $\val(\phi) < \delta$.
    \end{itemize}
\end{definition}

We represent $\NAEEkSAT$ constraints $C$ as sets of $k$ literals, and we use the notation $\vars(C)$ to denote the set of indices of variables involved in a $\NAEEthreeSAT$ constraint $C$.
For example, if $k = 3$ and $C := \set{v_1, \lnot v_3, \lnot v_4}$ then $\vars(C) = \set{1, 3, 4}$.
We define $\sgn(\ell) \in \pmo$ and $\var(\ell)$ to be the sign and the index of the variable underlying a literal $\ell$, respectively. For example, $\sgn(\lnot v_4) = -1$ and $\var(\lnot v_4) = 4$.

We next give a result showing $\NP$-hardness of $(15/16 + \eps, 1)$-$\NAEEthreeSAT$ for any constant $\eps > 0$ in \cref{thm:nae-e3-sat-hard}. This result follows by combining a tight hardness of approximation result for $\NAEEfourSAT$ by H\r{a}stad~\cite{hastad} (\cref{thm:hastad-4ary-csp}) with a reduction from approximate $\NAEEfourSAT$ to approximate $\NAEEthreeSAT$ (\cref{thm:nae-e4-sat-to-nae-e3-sat}). While \cref{thm:nae-e3-sat-hard} may very well be folklore, we did not find a suitable reference for it and so include its proof in \cref{sec:nae-e3-sat-hardness}.

\begin{theorem} \label{thm:nae-e3-sat-hard}
    For any constant \( \eps \in (0, 1/16) \), \( (15/16+\eps, 1) \)-$\NAEEthreeSAT$ is \( \NP \)-hard.
\end{theorem}

\section{Reduction from NAE3SAT to BinaryGapCRP}
\label{sec:reduction-nae3sat-to-bingapcrp}

\subsection{The Reduction}
\label{sec:the-reduction}

We start by giving our reduction from $(\delta, \eps)$-$\NAEEthreeSAT$ to $\gamma$-$\BinGapCRP_p$ for certain values of $\delta, \eps, \gamma, p$. In what follows, we will analyze the reduction and for which parameter values it is valid. We emphasize again that the reduction is similar to the $\NP$-hardness reduction in~\cite{manurangsi21} for $\LinDisc_{\infty}$.

Let \( \phi \) be the input $(\delta, \eps)$-$\NAEEthreeSAT$ formula, and assume that $\phi$ has $n$ variables \( v_1, \dots, v_n \) and $m$ constraints \( C_1, \dots, C_m \). We define the incidence matrix $\basis \in \R^{m \times n}$ of $\phi$ as
\begin{equation} \label{eq:def-B}
    B_{i,j} := B_{i,j}(\phi) = 
    \begin{cases}
        1 & \text{if constraint \( C_i \) contains literal \( v_j \) ,} \\
        -1 & \text{if constraint \( C_i \) contains literal \( \lnot v_j \) ,} \\
        0 & \text{otherwise .}
    \end{cases}
\end{equation}
We will also use the following gadget matrix, defined in \cite{manurangsi21}, which we analyze in \cref{sec:analysis-of-G}.
\begin{equation} \label{eq:def-G}
    \mat{G} := \begin{pmatrix}
        1 & 1 & -1 \\
        1 & -1 & 1 \\
        -1 & 1 & 1
        \end{pmatrix} \ \text{.}
\end{equation}
Let \( \deg_{\phi}(v_j) \) denote the number of constraints \( C_i \) in $\phi$ containing either \( v_j \) or $\lnot v_j$.
When the formula $\phi$ is clear from context, we simply write $\deg(v_j)$.
For $p \in [1, \infty]$, define the diagonal matrix
\begin{equation*} %
    \mat{D}_p := \mat{D}_p(\phi) = 
    \diag(\deg(v_1)^{1/p}, \ldots, \deg(v_n)^{1/p}) =
    \begin{pmatrix}
        \deg(v_1)^{1/p} & \\
        & \ddots & \\
        & & \deg(v_n)^{1/p} 
    \end{pmatrix} \in \R^{n \times n}
    \ \text{.}
\end{equation*}
Finally, we output the $\gamma$-$\BinGapCRP_p$ instance consisting of the matrix
\begin{equation} \label{eq:def-matA}
    \mat{A} := \mat{A}_p(\phi) = \begin{pmatrix}
        \frac{1}{3}\mat{B} & \frac{1}{3}\mat{B} & \frac{1}{3}\mat{B} \\
        & \mat{G} \otimes \mat{D}_p
    \end{pmatrix}
    \in \R^{(m+3n) \times 3n} \ \text{,}
\end{equation}
and distance threshold
\begin{equation} \label{eq:def-r}
r := (\eps m \cdot (4/3)^p + (1-\eps)m \cdot 2^p + 3m(2 + (4/3)^p))^{1/p} > 0\ \text{.}
\end{equation}

\subsection{Analysis of the matrix \texorpdfstring{$\mat{G}$}{G}}
\label{sec:analysis-of-G}

We begin by presenting some technical lemmas about $\mat{G}$ that will be useful in our hardness reductions.
We start with \cref{lem:gadget-properties}, which adapts \cite[Lemma 3]{manurangsi21} from the \( \ell_\infty \) norm to \( \ell_p \) norms.

\begin{lemma}\label{lem:gadget-properties}
    For any \( \vec{u} \in [0,1]^3 \) and \( b \in \pmo \), there exists \( \vec{z} \in \bit^3 \) such that
    \begin{enumerate}
        \item \label{item:G-lin-discrep} (Low discrepancy w.r.t. $\mat{G}$.) \( \norm{\mat{G}(\vec{u} - \vec{z})}_p^p \leq 2 + (4/3)^p \), where $\mat{G}$ is as defined in \cref{eq:def-G}.
        \item \label{item:sign-agree} (Sign agreement.) \( b \cdot (\vec{1}^T(\vec{u} - \vec{z} )) \geq 0 \).
        \item \label{item:ones-lin-discrep} (Low discrepancy w.r.t. $\vec{1}^T$.)
        \( \abs{\vec{1}^T(\vec{u} - \vec{z})} \leq 2 \).
    \end{enumerate}
\end{lemma}
\begin{proof}
    The second and third properties are identical to the corresponding properties in~\cite[Lemma 3]{manurangsi21} (which are proved there), and so we only prove the first property.
    We follow the structure of the proof of~\cite[Lemma 3, Item 1]{manurangsi21} to prove the first property. 

    We assume without loss of generality that $0 \leq u_1 \leq u_2 \leq u_3 \leq 1$ throughout the proof.
 This assumption is justified by the easy-to-check fact that for all $3 \times 3$ permutation matrices $\mat{P}$ and $\vec{x} \in \R^3$, $\norm{\mat{G}\vec{x}}_{\infty} = \norm{\mat{G}\mat{P}\vec{x}}_{\infty}$.
 
We now proceed to the proof of \cref{item:G-lin-discrep}. We will consider three cases, and in each case we will bound each coordinate of \( \mat{G}(\vec{u} - \vec{z}) \) individually.
    
 \noindent \textbf{Case 1}: \( \vec{1}^T \vec{u} \geq 2 \). In this case we set  \( \vec{z} := (0, 1, 1)^T \), and have
 \[
 \mat{G} (\vec{u} - \vec{z}) =
 \begin{pmatrix}
 u_1 + u_2 - u_3 \\
 u_1 - u_2 + u_3 \\
 -u_1 + u_2 + u_3 - 2
 \end{pmatrix} \ \text{.}
 \]
 We have following bounds:
    \begin{itemize}
        \item $-1 \leq -u_3 \leq (\mat{G}(\vec{u} - \vec{z}))_1 = (\mat{G}\vec{u})_1 \leq u_1 \leq 1$.
        \item \( -1 \leq -u_2 \leq (\mat{G}(\vec{u} - \vec{z}))_2 = (\mat{G} \vec{u}))_2 \leq u_3 \leq 1\).
        \item \( -4/3 \leq u_3 - 2 \leq (\mat{G}(\vec{u} - \vec{z}))_3 = (\mat{G}\vec{u})_3 - 2 \leq -u_1 \leq 0 \). 
    \end{itemize}

\noindent\textbf{Case 2}: \( \vec{1}^T \vec{u} < 2 \) and \( (\mat{G}\vec{u})_3 \leq 4/3 \).\footnote{In fact, the assumption that $\vec{1}^T \vec{u} < 2$ is not even necessary here.} In this case we set \( \vec{z} := \vec{0} \), and have
\[
 \mat{G} (\vec{u} - \vec{z}) =  
 \begin{pmatrix}
 u_1 + u_2 - u_3 \\
 u_1 - u_2 + u_3 \\
 -u_1 + u_2 + u_3
 \end{pmatrix} \ \text{.}
 \]
We have the following bounds:
    \begin{itemize}
        \item \( -1 \leq -u_3 \leq (\mat{G}(\vec{u} - \vec{z}))_1 = (\mat{G}\vec{u})_1 \leq u_1 \leq 1 \).
        \item \( -1 \leq -u_2 \leq (\mat{G}(\vec{u} - \vec{z}))_2  = (\mat{G}\vec{u})_2 \leq u_3 \leq 1\).
        \item \( 0 \leq u_3 \leq (\mat{G}(\vec{u} - \vec{z}))_3 = (\mat{G}\vec{u})_3 \leq 4/3 \).
    \end{itemize}

   \noindent \textbf{Case 3}: \( \vec{1}^T \vec{u} < 2 \) and \( (\mat{G}\vec{u})_3 > 4/3 \). In this case we set $\vec{z} := (0, 0, 1)^T$, and have
   \[
    \mat{G}(\vec{u} - \vec{z}) = \begin{pmatrix}
    u_1 + u_2 - u_3 + 1 \\
    u_1 - u_2 + u_3 - 1 \\
    -u_1 + u_2 + u_3 - 1
    \end{pmatrix} \ \text{.}
   \]
   We have the following bounds:
    \begin{itemize}
        \item 
        $ 0 \leq -u_3 + 1 \leq (\mat{G}(\vec{u} - \vec{z}))_1 = (\mat{G}\vec{u})_1 + 1 \leq u_1 + 1 = \frac{1}{2}(\vec{1}^T\vec{u} - (\mat{G}\vec{u})_3) + 1 < 1/3 + 1 = 4/3$.
        
        \item \( -1 \leq u_1 - 1 \leq (\mat{G}(\vec{u} - \vec{z}))_2 = (\mat{G}\vec{u})_2 - 1 \leq u_3 - 1 \leq 0 \).
       
        \item \( -1 \leq u_3 - 1 \leq (\mat{G}(\vec{u} - \vec{z}))_3 = (\mat{G}\vec{u})_3 - 1 \leq \vec{1}^T \vec{u} - 1 < 1 \).
    \end{itemize}

    In each case, we have \( \norm*{\mat{G}(\vec{u} - \vec{z})}^p_p \leq 2 + (4/3)^p \), as needed.
\end{proof}

\cref{lem:gadget-property-soundness} is the other part of the extension of~\cite[Lemma 3]{manurangsi21}, and facilitates a reduction to \( \BinGapCRP \), rather than just $\LinDisc$.
We both adapt~\cite[Lemma 3]{manurangsi21} to $\ell_p$ norms for $p < \infty$ and show that $\mat{G} (1/2 \cdot \vec{1}) = (1/2) \cdot \vec{1}$ is far from all vectors in
$\lat(\mat{G}) \setminus \set{\vec{0}, \vec{1}} = \mat{G} \cdot (\Z^3 \setminus \set{\vec{0}, \vec{1}})$ rather than just vectors in $\mat{G} \cdot (\bit^3 \setminus \set{\vec{0}, \vec{1}})$.

\begin{lemma}
    \label{lem:gadget-property-soundness}
    Let \( \vec{z} \in \Z^3 \) and $p \in [1, \infty]$. Then we have the following:
    \begin{enumerate}
        \item \label{item:G-soundness-01} For $\vec{z} \in \set{\vec{0}, \vec{1}}$, \( \norm{\mat{G}(1/2 \cdot \vec{1} - \vec{z})}^p_p = 3/2^p \) if $p < \infty$ and $\norm{\mat{G}(1/2 \cdot \vec{1} - \vec{z})}_{\infty} = 1/2$.
        \item \label{item:G-soundness-non01}
        For $\vec{z} \in \Z^3 \setminus \set{\vec{0}, \vec{1}}$,
        \( \norm{\mat{G}(1/2 \cdot \vec{1} - \vec{z})}^p_p \geq (2 + 3^p)/2^p \) if $p < \infty$ and $\norm{\mat{G}(1/2 \cdot \vec{1} - \vec{z})}_\infty \geq 3/2$.
    \end{enumerate}
\end{lemma}
\begin{proof}
    
    We note that \( \mat{G}((1/2) \cdot \vec{1} - \vec{z}) = (1/2) \cdot \vec{1} - \mat{G}\vec{z} \) so that\footnote{This is equivalent to the fact that \( \mat{G} \) has \( \vec{1} \) as an eigenvector with corresponding eigenvalue \( 1 \).}
    \[
    \set{(\mat{G}(1/2 \cdot \vec{1} - \vec{z}) : \vec{z} \in \Z^3} = \set{\mat{G}\vec{z} : \vec{z} \in \Z^3} + (1/2) \cdot \vec{1} \subseteq (\Z + 1/2)^3 \ \text{.}
    \]
    That is, for any $\vec{z} \in \Z^3$, each coordinate of $G((1/2 \cdot \vec{1} - \vec{z})$ is contained in $\Z + 1/2$. 
    Fix such a coefficient vector $\vec{z} \in \Z^3$.
    Then for $p < \infty$, $\norm{\mat{G}((1/2) \cdot \vec{1} - \vec{z})}_p^p = 3/2^p$ if $\mat{G} \vec{z} \in \bit^3$ and $\norm{\mat{G}((1/2) \cdot \vec{1} - \vec{z})}_p^p \geq (2 + 3^p)/2^p$ otherwise.
    Similarly, $\norm{\mat{G}((1/2) \cdot \vec{1} - \vec{z})}_{\infty} = 1/2$ if $\mat{G} \vec{z} \in \bit^3$ and $\norm{\mat{G}((1/2) \cdot \vec{1} - \vec{z})}_{\infty} \geq 3/2$ otherwise.
    The lemma then follows by noting that $\lat(\mat{G}) \cap \bit^3 = \set{\vec{0}, \vec{1}}$. \qedhere

\end{proof}

\subsection{Completeness Analysis}
\label{sec:completeness-analysis}

We start with an expression for the value of $\norm{\mat{A}(\vec{w} - \vec{x})}_p^p$ for the matrix $\mat{A} = \mat{A}(\phi)$ defined in \cref{eq:def-matA} in terms of a $\NAEEthreeSAT$ formula $\phi$ and vectors $\vec{w}, \vec{x}$.

\begin{lemma}
\label{lem:structural-identity}
    Let \( p \geq 1\), and for $i \in [3]$, let \( \vec{y}^i \in \R^n \), let \( \phi \) be an \( \NAEEthreeSAT \) formula with constraints \( C_1, \dots, C_m \), and let \( \mat{A} := \mat{A}_p(\phi) \) as defined in \cref{eq:def-matA}. Let \( \vec{y} := (\vec{y}^1, \vec{y}^2, \vec{y}^3 ) \in \R^{3n} \), let \( \vec{y}^\text{sum} := \vec{y}^1 + \vec{y}^2 + \vec{y}^3 \in \R^n \), and for \( j \in [n] \), let \( \vec{y}_j := (y^1_j, y^2_j, y^3_j) \in \R^3 \).
    Then 
    \begin{equation}\label{eqn:structural-identity-np}
        \norm{\mat{A}\vec{y}}_p^p 
        = (1/3^p) \cdot \norm*{\mat{B}\vec{y}^\text{sum}}_p^p 
        + \sum_{j=1}^n \deg(v_j) \cdot \norm*{\mat{G} \vec{y}_j}_p^p \ \text{.}
    \end{equation}    
    
    Let \( \vec{x}^1, \vec{x}^2, \vec{x}^3 \in \R^n \), let \( \vec{x} := (
    \vec{x}^1, \vec{x}^2, \vec{x}^3) \in \R^{3n} \), and let \( \vec{x}^\text{sum} := \vec{x}^1 + \vec{x}^2 + \vec{x}^3 \in \R^{n} \). Furthermore, for $j \in [n]$, let \( \vec{x}_j := (x^1_j, x^2_j, x^3_j)^T \in \R^3 \), and let \( \vec{b}_i \) denote the \( i \)th row of \( \mat{B} \). Then in particular,
    \begin{equation}\label{eqn:structural-identity-instantiation}
        \norm*{\mat{A}\parens*{\frac{1}{2} \cdot \vec{1} - \vec{x}}}_p^p = 
        \sum_{i = 1}^m \parens*{\abs*{\frac{1}{3} \vec{b}_i \parens*{\frac{3}{2} \cdot \vec{1} - \vec{x}^\text{sum}}}^p + \sum_{j \in \vars(C_i)}\norm*{\mat{G}\parens*{\frac{1}{2} \cdot \vec{1} -        \vec{x}_j
        }}_p^p} \ \text{.}
    \end{equation}
\end{lemma}
\begin{proof}
    For \cref{eqn:structural-identity-np},     
    \begin{align*}
        \norm{\mat{A}(\vec{y})}_p^p 
        &= \norm*{(1/3) \cdot (\mat{B}\vec{y}^1 + \mat{B}\vec{y}^2 + \mat{B}\vec{y}^3)}_p^p + \norm{(\mat{G} \otimes \mat{D}_p) \cdot \vec{y}}_p^p \\
        &= 
        \norm*{(1/3) \cdot \mat{B}\vec{y}^\text{sum}}_p^p 
        + \sum_{j=1}^n \deg(v_j) \cdot \norm*{\mat{G} \vec{y}_j}_p^p \ \text{.}
    \end{align*}
    For \cref{eqn:structural-identity-instantiation}, we set \( \vec{y} := \frac{1}{2} \cdot \vec{1} - \vec{x} \) and continue by re-grouping some terms.
    \begin{align*}
        \norm*{\mat{A}\parens*{\frac{1}{2} \cdot \vec{1} - \vec{x}}}_p^p 
        &= \frac{1}{3^p} \cdot \norm*{\mat{B}\parens*{\frac{3}{2} \cdot \vec{1} - \vec{x}^\text{sum}}}_p^p 
        + \sum_{j=1}^n \deg(v_j) \cdot \norm*{\mat{G}\parens*{\frac{1}{2} \cdot \vec{1} - \vec{x}_j}}_p^p \\
        &= \frac{1}{3^p} \cdot \sum_{i = 1}^m \abs*{\vec{b}_i \parens*{\frac{3}{2} \cdot \vec{1} - \vec{x}^\text{sum}}}^p + \sum_{j=1}^n \deg(v_j) \cdot \norm*{\mat{G}\parens*{\frac{1}{2} \cdot \vec{1} - \vec{x}_j}}_p^p \\
        &= \sum_{i = 1}^m \parens*{\abs*{\frac{1}{3} \vec{b}_i \parens*{\frac{3}{2} \cdot \vec{1} - \vec{x}^\text{sum}}}^p + \sum_{j \in \vars(C_i)}\norm*{\mat{G}\parens*{\frac{1}{2} \cdot \vec{1} - \vec{x}_j}}_p^p} \text{.}
    \end{align*}
    The first equality follows by setting $\vec{w} := 1/2 \cdot \vec{1}$ in \cref{eqn:structural-identity-np}.
    The third equality follows by noting for any expression $K(j)$ depending on $j$,
    \[\sum_{j = 1}^n \deg(v_j) K(j) = \sum_{j = 1}^n \sum_{\vars(C_i) : x_j \in \vars(C_i)} K(j) = \sum_{i=1}^m \sum_{j: x_j \in \vars(C_i)} K(j) \ \text{.}
    \]
    Indeed, one can view the two double sums as counting constraint-variable incidences in two different ways.
\end{proof}

We now state lemmas that will be used in the proof of \cref{thm:nae3sat-to-bingapcrp}. We begin with a lemma for completeness in which we upper bound $\lindisc_p(\mat{A})$.

\begin{lemma}\label{lem:correctness-completeness}
    Let \( p \in [1, \infty) \), let $\eps \in [0, 1]$, let \( \phi \) be an $\NAEEthreeSAT$ instance with $m$ constraints and $\val(\phi) \geq \eps$, and let \( \mat{A} := \mat{A}_p(\phi) \) be defined as in \cref{eq:def-matA}.  Then
    \begin{equation}\label{eq:lindisc-completeness}
        \lindisc_p(\mat{A})^p \leq (\eps \cdot \parens*{4/3}^p + (1-\eps) \cdot 2^p + 3\parens*{2+\parens*{4/3}^p}) \cdot m \ \text{.}
    \end{equation}
    \end{lemma}
    \begin{proof}
        Let $\psi$ be an assignment to the variables of $\phi$ such that $\val_{\psi}(\phi) \geq \eps$.
        Let \( \vec{w}^1, \vec{w}^2, \vec{w}^3 \in [0,1]^n \), and let \( \vec{w} := (\vec{w}^1, \vec{w}^2, \vec{w}^3) \). We define \( \vec{x} := (\vec{x}^1, \vec{x}^2, \vec{x}^3) \) for \( \vec{x}^1, \vec{x}^2, \vec{x}^3 \in \bit^n \) 
        as follows.
        For each \( j \in [n] \), the vector \( \vec{x}_j := (x^1_j, x^2_j, x^3_j)^T \) is the vector $\vec{z}$ guaranteed by \cref{lem:gadget-properties} applied to the vector \( \vec{u} := \vec{w}_j = (w^1_j, w^2_j, w^3_j)^T \) and sign \( b := 1 - 2 \psi(v_j) \). 

    We begin by upper bounding \( \norm{\mat{A}(\vec{w} - \vec{x})}_p^p \), which we do by upper bounding the terms on the right-hand side of \cref{eqn:structural-identity-np} individually with \( \vec{y} := \vec{w} - \vec{x} \in \R^{3n} \). 
    From \cref{lem:gadget-properties}, \cref{item:G-lin-discrep} and the fact that $\phi$ has total degree $\sum_{j = 1}^n \deg(v_j) = 3m$, we have 
    \begin{equation} \label{eqn:downstairs-upper-bound}
    \sum_{j=1}^n \deg(v_j) \cdot \norm*{\mat{G}\parens*{\vec{w}_j - \vec{x}_j}}_p^p \leq \sum_{j = 1}^n \deg(v_j) \cdot \parens*{2 + (4/3)^p} = 3m(2 + (4/3)^p) \text{.}
    \end{equation}
    
    To upper bound \( \norm*{\frac{1}{3}\mat{B}(\vec{w}^\text{sum} - \vec{x}^\text{sum})}^p_p \), we upper bound the magnitude of each coordinate $i$ of \( \mat{B}(\vec{w}^\text{sum} - \vec{x}^\text{sum}) \in \R^m \). 
    Suppose that constraint $C_i = \set{\ell_{i,1}, \ell_{i,2}, \ell_{i,3}}$ for some literals $\ell_{i,1}, \ell_{i,2}, \ell_{i,3}$. Then
    \begin{equation} \label{eq:Bwsum-xsum-i}
     (\mat{B}(\vec{w}^\text{sum} - \vec{x}^\text{sum}))_i
     = \sum_{j=1}^3 \sgn(\ell_{i,j}) \cdot (w_{\var(\ell_{i,j})}^\text{sum} - x_{\var(\ell_{i,j})}^\text{sum}) \ \text{.}
    \end{equation}

    By \cref{lem:gadget-properties}, \cref{item:ones-lin-discrep}, we have that $\abs{\vec{1}^T (\vec{w}_j - \vec{x}_j)} = \abs{w^\text{sum}_j - x^\text{sum}_j} \leq 2$ for all $j \in [n]$.
    So, the expression in \cref{eq:Bwsum-xsum-i} always has magnitude at most $6$ by triangle inequality and has magnitude at most $4$ as long as the signs of its summands are not all the same (i.e., as long as $\sgndiff(\vec{b}_i, \vec{w}^{\text{sum}} - \vec{x}^{\text{sum}})$ holds, where $\vec{b}_i$ is the $i$th row of $\basis$).
    We show that this is the case if $\psi$ satisfies constraint $C_i$.
    Assume without loss of generality that $w_{\var(\ell_{i,j})}^\text{sum} \neq x_{\var(\ell_{i,j})}^\text{sum} $ for $j \in [3]$ (since otherwise the expression clearly has magnitude at most $4$).
    By \cref{lem:gadget-properties},
    \cref{item:sign-agree} and our choice of the sign $1 - 2 \psi(v_j)$ when applying \cref{lem:gadget-properties} to $\vec{w}_j$,
    we have that for all $j \in [n]$, $(1 - 2\psi(v_j)) \cdot (w^\text{sum}_j - x^\text{sum}_j) \geq 0$, which implies that $(1 - 2\psi(v_j)) = \sgn( w^\text{sum}_j - x^\text{sum}_j )$.
    So, for $j \in [3]$,
    \begin{equation} \label{eqn:sign-char}
    \sgn(\ell_{i,j}) \cdot \sgn(w_{\var(\ell_{i,j})}^\text{sum} - x_{\var(\ell_{i,j})}^\text{sum})
    = \sgn(\ell_{i,j}) \cdot (1 - 2\psi(v_j))
    =  \begin{cases}
            -1 & \text{if $\psi$ satisfies literal $\ell_{i,j}$ ,} \\
            1 & \text{otherwise .}
        \end{cases}
    \end{equation}

By definition, a $\NAEEthreeSAT$ constraint $C_i$ is satisfied if and only if at least one of its literals is satisfied and at least one is not satisfied. So, if $\psi$ satisfies $C_i$,
$\sgn(\ell_{i,j}) \cdot \sgn(w_{\var(\ell_{i,j})}^\text{sum} - x_{\var(\ell_{i,j})}^\text{sum})$
is equal to $1$ for some $j \in [3]$ and is equal to $-1$ for some $j \in [3]$.
It follows that the expression in \cref{eq:Bwsum-xsum-i} has magnitude at most $4$.
Because $\val(\phi) \geq \eps$ by assumption, 
\begin{equation} \label{eqn:upstairs-upper-bound}
\norm{(1/3) \cdot \basis (\vec{w}^{\text{sum}} - \vec{x}^{\text{sum}})}_p^p
\leq \frac{1}{3^p} \cdot (\eps m \cdot 4^p + (1 - \eps) m \cdot 6^p)
= (\eps \cdot (4/3)^p + (1 - \eps) \cdot 2^p) \cdot m \ \text{.}
\end{equation}

    Combining the upper bounds of \cref{eqn:downstairs-upper-bound,eqn:upstairs-upper-bound} to the terms in \cref{eqn:structural-identity-np} (with \( \vec{y} := \vec{w} - \vec{x} \)) implies \cref{eq:lindisc-completeness}, as needed.
\end{proof}

\subsection{Soundness Analysis}
\label{sec:soundness-analysis}

We now give two lemmas for analyzing the soundness of our reduction. The first lemma establishes that vectors of the form \( \mat{A}\vec{x} \) where \( \vec{x} \in \bit^{3n}\) are closest vectors in $\lat(\mat{A})$ to \( \mat{A} \cdot (1/2 \cdot \vec{1}) \).

\begin{lemma}
\label{lem:correctness-soundness-1}
    Let $p \in [1, \infty]$, let \( \phi \) be an $\NAEEthreeSAT$ instance, and let \( \mat{A} := \mat{A}(\phi) \) be as defined in \cref{eq:def-matA}.
    Then for every \( \vec{y} \in \Z^{3n} \), there exists \( \vec{x} := (\vec{x}', \vec{x}', \vec{x}') \in \bit^{3n} \) with \( \vec{x}' \in \bit^n \) such that \( \norm{\mat{A}(1/2 \cdot \vec{1} - \vec{x})}_p \leq \norm{\mat{A}(1/2 \cdot \vec{1} - \vec{y})}_p \).
\end{lemma}
\begin{proof}
    Let $\vec{y}^1, \vec{y}^2, \vec{y}^3 \in \Z^n$ be such that \( \vec{y} = (\vec{y}^1, \vec{y}^2, \vec{y}^3) \in \Z^{3n} \), let \( \vec{y}^\text{sum} := \vec{y}^1 + \vec{y}^2 + \vec{y}^3 \in \Z^n \), and for each \( j \in [n] \), let \( \vec{y}_j := (y^1_j, y^2_j, y^3_j) \). 
    For each \( j \in [n] \), define
    \begin{equation}
        \label{eqn:setting-x-from-y}
        x_j' := 
        \begin{cases}
            0 & \text{if \( y^\text{sum}_j \leq 1 \) ,} \\
            1 & \text{if \( y^\text{sum}_j \geq 2 \) .}
        \end{cases}
    \end{equation}
Furthermore, define $\vec{x}' := (x_1', \ldots, x_n')^T \in \bit^n$, $\vec{x} := (\vec{x}', \vec{x}', \vec{x}') \in \bit^{3n}$, and for each \( j \in [n] \), \( \vec{x}_j := (x'_j, x'_j, x'_j ) \in \bit^3 \).
    
    We continue the proof by upper and lower bounding the $i$th terms in the sum in \cref{eqn:structural-identity-instantiation} for \( \norm{\mat{A}(1/2 \cdot \vec{1} - \vec{x})}^p_p \) and for \( \norm{\mat{A}(1/2 \cdot \vec{1} - \vec{y})}^p_p \), respectively. Specifically, we show that the value of the $i$th term in the former sum is at most the value of the $i$th term in the latter sum for every $i \in [m]$.
    By the definition of $\vec{x}$ in terms of $\vec{y}$, these $i$th terms are the same if $\vec{y}_j \in \set{\vec{0}, \vec{1}}$ for all $j \in \vars(C_i)$, where $C_i$ is the $i$th constraint of $\phi$.
    So, we assume without loss of generality that this is not the case, i.e., that for every $i \in [m]$ there exists $j \in \var(C_i)$ such that $\vec{y}_j \in \Z^3 \setminus \set{\vec{0}, \vec{1}}$.
    
    We begin by upper bounding the terms in \( \norm{\mat{A}(1/2 \cdot \vec{1} - \vec{x})}^p_p \). Fix some \( i \in [m] \) corresponding to a term of \cref{eqn:structural-identity-instantiation}.
    Because \( \vec{x}_j \in \set{\vec{0}, \vec{1}} \) for all \( j \in [n] \) by definition, by \cref{lem:gadget-property-soundness}, we have the equality
    \begin{equation}
        \label{eqn:downstairs-x}
        \sum_{j \in \vars(C_i)} \norm{\mat{G}(1/2 \cdot \vec{1} - \vec{x}_j)}^p_p = \sum_{j \in \vars(C_i)} 3/2^p = 3 \cdot (3/2^p) = 9/2^p \  \text{.}
    \end{equation}
    
    Let $\vec{b}_i$ be the $i$th row of the matrix $\basis$ in \cref{eq:def-B}, as in \cref{eqn:structural-identity-instantiation}.
    Again using the fact that each $\vec{x}_j \in \set{\vec{0}, \vec{1}}$, we have
    \begin{equation}
        \label{eqn:upstairs-x}
        \abs{\iprod{\vec{b}_i, (1/2) \cdot \vec{1} - \vec{x}'}}
        = \Big|\sum_{j \in \vars(C_i)} B_{i,j} \cdot (1/2 - x_j') \Big|
        = \begin{cases}
            1/2 & \text{if $\sgndiff(\vec{b}_i, 1/2 \cdot \vec{1} - \vec{x}')$ ,} \\
            3/2 & \text{otherwise .}
        \end{cases}
    \end{equation}
    Combining \cref{eqn:downstairs-x,eqn:upstairs-x}, we have 
    \begin{equation}
        \label{eqn:total-bound-x}
                \abs{\iprod{\vec{b}_i, (1/2) \cdot \vec{1} - \vec{x}'}}^p + \sum_{j \in C_i}\norm{\mat{G}\parens{1/2 \cdot \vec{1} - \vec{x}_j}'}_p^p  = 
        \begin{cases}
            10/2^p & \text{if $\sgndiff(\vec{b}_i, 1/2 \cdot \vec{1} - \vec{x}')$ ,}  \\
            (9 + 3^p)/2^p  & \text{otherwise .}
        \end{cases}
    \end{equation}

    We now lower bound \( \norm{\mat{A}(1/2 \cdot \vec{1} - \vec{y})}^p_p \). We have that
    \begin{equation}
        \label{eqn:downstairs-y}
        \sum_{j \in \vars(C_i)} \norm{\mat{G}(1/2 \cdot \vec{1} - \vec{y}_j)}^p_p 
        \geq (3/2)^p + 8/2^p \ \text{.}
    \end{equation}
    Indeed, by \cref{lem:gadget-property-soundness}, the terms for which \( \vec{y}_j \in \set{\vec{0}, \vec{1}} \) are equal to $3/2^p$ and the terms in the sum for which $\vec{y}_j \notin \set{\vec{0}, \vec{1}}$ must have value at least $(3/2)^p + 2/2^p$. 
    And, since we are only considering constraints $C_i$ for which there exists \( j \) such that \( \vec{y}_j \notin \set{\vec{0}, \vec{1}} \), at least one term must have value at least $(3/2)^p + 2/2^p$.
    
    We continue by giving a lower bound for \( \abs{(1/3) \cdot \vec{b}_i(3/2 \cdot \vec{1} - \vec{y}^\text{sum})} \). Here, we observe that \( (1/2) - (y^\text{sum}_j / 3) \in \set{1/6,3/6,5/6} + \Z \). This implies that
    \begin{align}
        \label{eqn:upstairs-y}
        \iprod{\vec{b}_i, (1/2) \cdot \vec{1} - \vec{y}^{\text{sum}}/3}
        &= \Big|\sum_{j \in \vars(C_i)} B_{i,j} \cdot (1/2 - y_j^{\text{sum}}/3) \Big| \nonumber \\
        &\geq \begin{cases}
            1/6 & \text{if $\sgndiff(\vec{b}_i, 1/2 \cdot \vec{1} - \vec{y}^{\text{sum}}/3)$} \\
            1/2 & \text{otherwise}
        \end{cases}
        \text{,}
    \end{align}
    where the inequality follows by noting that because all three summands in the middle expression are elements of \( \set{1/6, 3/6, 5/6} + \Z \). If all signs of the summands are the same, one obtains the lower bound \( 3/6 = 1/2 \), whereas if the signs of the summands differ, one only has the lower bound of \( 1/6 \).
        
    By combining \cref{eqn:downstairs-y,eqn:upstairs-y}, we have the lower bound
    \begin{align}
        \label{eqn:total-bound-y}
       \abs{\iprod{\vec{b}_i, (1/2) \cdot \vec{1} - \vec{y}^{\text{sum}}/3}}^p + & \sum_{j \in \vars(C_i)} \norm{\mat{G}(1/2 \cdot \vec{1} - \vec{y}_j)}^p_p \nonumber
        \\ &\geq \begin{cases}
            1/6^p + (3/2)^p + 8/2^p & \text{if $\sgndiff(\vec{b}_i, 1/2 \cdot \vec{1} - \vec{y}^{\text{sum}}/3)$}  \\
           (9 + 3^p)/2^p & \text{otherwise}
        \end{cases}
        \text{.}
    \end{align}
    
    Finally, by definition of \( \vec{x}' \) (\cref{eqn:setting-x-from-y}), we have that for all \( j \in [n] \), $\sgn(x_j' - 1/2) = \sgn(y_j^{\text{sum}}/3 - 1/2)$.
    Then, it follows that for all \( i \in [m] \) that $\sgndiff(\vec{b}_i, 1/2 \cdot \vec{1} - \vec{x}')$ if and only if $\sgndiff(\vec{b}_i, 1/2 \cdot \vec{1} - \vec{y}^{\text{sum}}/3)$.
    By noting that for all \( p \geq 1 \), 
    \( 2/2^p \leq 1/6^p + (3/2)^p \), we combine \cref{eqn:total-bound-x,eqn:total-bound-y} to obtain that for all $i \in [m]$,
    \begin{align*}
        & \abs{\iprod{\vec{b}_i, (1/2) \cdot \vec{1} - \vec{x}'}}^p + \sum_{j \in \vars(C_i)}\norm{\mat{G}\parens{1/2 \cdot \vec{1} - \vec{x}_j}}_p^p \\
        \leq \ &  \abs{\iprod{\vec{b}_i, (1/2) \cdot \vec{1} - \vec{y}^{\text{sum}}/3}}^p +\sum_{j \in \vars(C_i)} \norm{\mat{G}(1/2 \cdot \vec{1} - \vec{y}_j)}^p_p \ \text{,}
    \end{align*}
    where we note that the left- and right-hand sides of the inequality are equal to the $i$th term in the sum in \cref{eqn:structural-identity-instantiation}. The result follows.
\end{proof}

The following lemma shows soundness for our reduction. Specifically, it shows how to construct an assignment satisfying a large fraction of constraints in an $\NAEEthreeSAT$ formula $\phi$ from a close vector $\mat{A} \vec{x}$ for $\vec{x} \in \bit^{3n}$ to $\mat{A} (1/2 \cdot \vec{1}_{3n})$, where $\mat{A} = \mat{A}_p(\phi)$.
\begin{lemma}\label{lem:correctness-soundness-2}
    Let $p \in [1, \infty)$, let \( \phi \) be an $\NAEEthreeSAT$ instance with $n$ variables and $m$ constraints, let \( \mat{A} := \mat{A}_p(\phi) \) be as defined in \cref{eq:def-matA}, and let \( \delta \in [0,1] \). 
    Suppose that there exists some vector \( \vec{x} := (\vec{x}^1, \vec{x}^2, \vec{x}^3) \in  \Z^{3n} \) such that
    \begin{equation}
        \label{eqn:soundness-upper-bound}
        \norm{\mat{A} ((1/2) \cdot \vec{1} - \vec{x})}_p^p \leq \delta m \cdot (10/2^p) + (1-\delta)m \cdot (9/2^p + (3/2)^p) \ \text{.}    
    \end{equation}
    Then $\val(\phi) \geq \delta$.
\end{lemma}
\begin{proof}

    Suppose that there exists \( \vec{x} \in \Z^{3n} \) satisfying \cref{eqn:soundness-upper-bound}.
    Then by \cref{lem:correctness-soundness-1}, we may assume without loss of generality that
    \( \vec{x} = (\vec{x}', \vec{x}', \vec{x}') \) for \( \vec{x}' \in \bit^n \) and let \( \vec{x}_j := (x'_j, x'_j, x'_j ) \) for each \( j \in [n] \).
    Therefore, for each \( i \in [m] \), both \cref{eqn:downstairs-x,eqn:upstairs-x} hold for \( \vec{x} \).
     
     First, we obtain the bound
    \begin{align*}
    (1/3^p) \cdot \norm{\mat{B}((3/2) \cdot \vec{1} - \vec{x}^\text{sum})}^p_p &= \norm{\mat{A} ((1/2) \cdot \vec{1} - \vec{x})}_p^p - \sum_{j=1}^n \deg(v_j) \cdot \norm{\mat{G}((1/2) \cdot \vec{1} - \vec{x}_j)}^p_p \\
    &= \norm{\mat{A} ((1/2) \cdot \vec{1} - \vec{x})}_p^p - m \cdot (9/2^p) \\
    &\leq \delta m \cdot (1/2^p) + (1 - \delta) m \cdot (3/2)^p \ \text{,}
    \end{align*}
    where the first equality uses \cref{eqn:structural-identity-np} with \( \vec{y} := (1/2) \cdot \vec{1} - \vec{x} \), the second equality uses \cref{eqn:downstairs-x}, and the inequality is by assumption (\cref{eqn:soundness-upper-bound}).

    It then follows from \cref{eqn:upstairs-x} that at least an \( \delta \) fraction of indices \( i \in [m] \) satisfy $\sgndiff(\vec{b}_i, 1/2 \cdot \vec{1} - \vec{x}')$.
    Define an assignment $\psi$ to the variables $v_1, \ldots, v_n$ of $\phi$ by $\psi(v_j) = x_j'$.
    We claim that if $\sgndiff(\vec{b}_i, 1/2 \cdot \vec{1} - \vec{x}')$ then $\psi$ satisfies constraint $C_i$, which combined with the fact that at least an \( \delta \) fraction of indices \( i \in [m] \) satisfy $\sgndiff(\vec{b}_i, 1/2 \cdot \vec{1} - \vec{x}')$, would imply the lemma.

    It remains to prove the claim.
    Observe that if constraint $C_i = \set{\ell_{i,1}, \ell_{i,2}, \ell_{i,3}}$ 
    for literals $\ell_{i,1}, \ell_{i,2}, \ell_{i,3}$ then, as in
    \cref{eq:Bwsum-xsum-i},
    \[
     (\mat{B}(1/2 \cdot \vec{1} - \vec{x}'))_i
     = \sum_{j=1}^3 \sgn(\ell_{i,j}) \cdot (1/2 - x_{\var(\ell_{i,j})}) \ \text{.}
    \]    
    Furthermore, as in \cref{eqn:sign-char}, for $j \in [3]$,
    \[
   \sgn(\ell_{i,j}) \cdot \sgn(1/2 - x_{\var(\ell_{i,j})})
    =  \begin{cases}
            -1 & \text{if $\psi$ satisfies literal $\ell_{i,j}$ ,} \\
            1 & \text{otherwise .}
        \end{cases}
    \]
    The predicate $\sgndiff(\vec{b}_i, 1/2 \cdot \vec{1} - \vec{x}')$ holds exactly when $ \sgn(\ell_{i,j}) \cdot (1/2 - x_{\var(\ell_{i,j})}) = 1$ for some $j \in [3]$ and $\sgn(\ell_{i,j}) \cdot (1/2 - x_{\var(\ell_{i,j})}) = -1$ for some $j \in [3]$, meaning that $\psi$ satisfies at least one literal in $C_i$ and does not satisfy at least one literal in $C_i$. This in turn implies that $\psi$ satisfies $C_i$, as needed.

\end{proof}

\subsection{The Main Theorem}
\label{sec:main-theorem}

We now state and prove the main technical theorem of the paper. 
\begin{theorem} \label{thm:nae3sat-to-bingapcrp}
    Let $\delta, \eps \in \R$ be such that $0 < \delta \leq \eps \leq 1$ and let $p \in [1, \infty)$. There is a polynomial-time mapping reduction from $(\delta, \eps)$-$\NAEEthreeSAT$ to $\gamma$-$\BinGapCRP_p$ for  
\begin{equation} \label{eq:hardness-ratio}
    \gamma = \gamma(\delta, \eps, p) :=
    \Big( \frac{\delta \cdot 10/2^p + (1 - \delta) \cdot ((3/2)^p + 9/2^p)}{\eps \cdot (4/3)^p + (1-\eps)2^p + 3(2 + (4/3)^p)}\Big)^{1/p} \ \text{.}
    \end{equation}
\end{theorem}

\begin{proof}
    Let \( \phi \) be the input instance of \( (\delta, \eps) \)-$\NAEEthreeSAT$. Assume that $\phi$ has \( m \) constraints.
    The reduction outputs the $\BinGapCRP_p$ instance consisting of the matrix \( \mat{A} = \mat{A}_p(\phi) \) defined in \cref{eq:def-matA} and the distance threshold \( r := (\eps m \cdot (4/3)^p + (1-\eps)m \cdot 2^p + 3m(2 + (4/3)^p))^{1/p} \) defined in \cref{eq:def-r}.

    It is clear that the reduction runs in polynomial time, and it remains to show correctness. We start with completeness. If \( \phi \) is a YES instance of $(\delta, \eps)$-$\NAEEthreeSAT$, then $\lindisc_p(\mat{A}) \leq r$ by \cref{lem:correctness-completeness} and so the output is a YES instance of $\BinGapCRP_p$.
    
    For soundness, we show the contrapositive. Suppose that  \[ \mu_p(\lat(\mat{A})) \leq \gamma r = (\delta m \cdot 10/2^p + (1 - \delta)m \cdot ((3/2)^p + 9/2^p))^{1/p} \ \text{.} \]
    Then 
    there exists $\vec{y} \in \Z^{3n}$ such that \(\norm{\mat{A}((1/2 \cdot \vec{1}) - \vec{y})}_p \leq \gamma r \), and so $\val(\phi) \geq \delta$ by \cref{lem:correctness-soundness-2}. The result follows.

\end{proof}

We conclude by obtaining our main theorem (\cref{thm:intro-np-hardness}), which asserts that $\gamma$-$\BinGapCRP_p$ is $\NP$-hard for any constant $1 \leq \gamma < \gamma(p)$, as a corollary of \cref{thm:nae3sat-to-bingapcrp}.

\begin{proof}[Proof of \cref{thm:intro-np-hardness}]
Let $\gamma(p) := \gamma(15/16, 1, p)$, where $\gamma(\cdot, \cdot, \cdot)$ is as defined in \cref{eq:hardness-ratio}.
Then
\begin{align*}
\gamma(p)
&= \Big( \frac{(15/16) \cdot 10/2^p + (1/16) \cdot ((3/2)^p + 9/2^p)}{(4/3)^p + 3(2 + (4/3)^p)}\Big)^{1/p} \\
&= \Big(\frac{9^p + 159 \cdot 3^p}{8^{p+2} + 96 \cdot 6^p}\Big)^{1/p} \ \text{,} 
\end{align*}

By \cref{thm:nae-e3-sat-hard}, \( (15/16+\eps', 1) \)-NAE-E3-SAT is \( \NP \)-hard for every constant \( \eps' > 0 \).
So, by \cref{thm:nae3sat-to-bingapcrp}, $\gamma(15/16+\eps', 1, p)$-$\BinGapCRP$ is $\NP$-hard for every such $\eps'$.
The theorem follows by noting that for every constant $\eps > 0$ there exists a constant $\eps' > 0$ such that $\gamma(15/16+\eps', 1, p) \geq \gamma(p) - \eps$. %
\end{proof}

\section{\texorpdfstring{\( \Pi_2 \)}{Pi\_2}-hardness of approximating BinaryGapCRP in the \texorpdfstring{\( \ell_\infty\)}{ell\_infty} norm}

In this section, we show that the $\Pi_2$-hardness of approximation result for $\LinDisc_{\infty}$ given in Manurangsi~\cite{manurangsi21} can be extended to \( \BinGapCRP_{\infty} \).
We first give the definition of $\AENAEEthreeSAT$, the $\Pi_2$-hard problem that we reduce to approximate $\BinGapCRP$. 
$\AENAEthreeSAT$ (in which constraints may have either two or three literals) was originally shown to be $\Pi_2$-hard in \cite{EiterGottlob1995Eigenvector}, and (a very special case of) $\AENAEEthreeSAT$ was shown to be $\Pi_2$-hard in \cite[Theorem 4.7]{docker-AENAE-E3-SAT}.

\begin{definition}[$\AENAEEthreeSAT$]
    The $\AENAEEthreeSAT$ problem is the decision problem where, on input a $\NAEEthreeSAT$ formula $\phi(V_A, V_E)$ defined over the disjoint union $V_A \sqcup V_E$ of variable sets $V_A$ and $V_E$, the goal is to distinguish between the following two cases. 
    \begin{itemize}
        \item (YES instance.) For every assignment \( \psi_A  : V_A \to \bit \), there exists an assignment \( \psi_E : V_E \to \bit \) such that $(\psi_A, \psi_E)$ satisfies $\phi$.
        
        \item (NO instance.) There exists an assignment \( \psi_A : V_A \to \bit \) such that for all assignments \( \psi_E : V_E \to \bit )\), $(\psi_A, \psi_E)$ does not satisfy $\phi$.
    \end{itemize}
\end{definition}

Let $\phi(V_A, V_E)$ be a $\NAEEthreeSAT$ formula with \( n' = \abs{V_A} \) and \(  n - n' = \abs{V_E} \).
We define the matrix $\mat{A}'= \mat{A}'(\phi)$ (originally defined in~\cite{manurangsi21}) as
\begin{equation}
    \label{eqn:def-matAprime}
    \mat{A}' := 
    \left(\begin{array}{c|c|c|c|c|c|c}
      \multicolumn{6}{c|}{\mat{A}_{\infty}(\phi)} & \mat{0}_{n' \times n'} \\
      \hline 
      \frac{2}{3} \mat{I}_{n'} & \mat{0}_{n' \times (n - n')} & \frac{2}{3} \mat{I}_{n'} & \mat{0}_{n' \times (n - n')} & \frac{2}{3} \mat{I}_{n'} & \mat{0}_{n' \times (n - n')} & -2\mat{I}_{n'} \\
      \hline
      \multicolumn{6}{c|}{\mat{0}_{n' \times 3n}} & \frac{8}{3}\mat{I}_{n'}  \\
    \end{array}\right)
    \in \R^{(m + 3n + 2n') \times (3n + n')} \ \text{,}
\end{equation}
where $ \mat{A}_{\infty}(\phi)$ is as defined in \cref{eq:def-matA} and where we have used the notation $\mat{0}_{m \times n}$ to denote an $m \times n$ matrix of $0$s.

We next give a simple expression for $\norm{\mat{A}'(\vec{y}')}_\infty$ for a vector $\vec{y}'$ in terms of the blocks of $\mat{A}'$.
\begin{lemma}[{\cite[Equation 4]{manurangsi21}}]
    \label{lem:pi2-structural-identity}
    
    Let \( \phi(V_A, V_E) \) be an $\AENAEEthreeSAT$ instance with \( n' \) universal variables (variables in \( V_A \)) and \( n - n' \) existential variables (variables in \( V_E \)), and for each \( i \in [3] \), let \( \vec{y}^i \in \R^n \), and let \( \vec{y}^* \in \R^{n'} \). Furthermore, let \( \vec{y}' := (\vec{y}^1, \vec{y}^2, \vec{y}^3, \vec{y}^*) = (\vec{y}, \vec{y}^*) \in \R^{3n + n'} \). Finally, let \( \vec{y}^\text{sum} = \vec{y}^1 + \vec{y}^2 + \vec{y}^3 \). Then $\mat{A}' = \mat{A}'(\phi)$ satisfies
    \begin{equation}\label{eqn:structural-identity-pi2}
    \norm{\mat{A}'\vec{y}'}_\infty = 
    \max \set{\norm*{\mat{A}\vec{y}}_\infty, \max_{i \in [n']}\abs{(2/3) \cdot y^\text{sum}_i - 2 y^*_i}, (8/3) \cdot  \max_{i \in [n']} 
    \abs{y^*_i}}
\end{equation}
\end{lemma}
\begin{proof}
The claim follows by inspection of the blocks of $\mat{A}'\vec{y}'$.
\end{proof}

We next state the completeness and soundness results from the reduction in \cite{manurangsi21} for showing $\Pi_2$-hardness of $\lindisc_{\infty}(A)$.
\begin{lemma}[Completeness for $\Pi_2$-hardness reduction, \cite{manurangsi21}] \label{lem:completeness-Pi2-manurangsi}
    Let \( \phi \) be an $\AENAEEthreeSAT$ instance, and \( \mat{A}' := \mat{A}'(\phi) \) as defined in \cref{eqn:def-matAprime}. Suppose that \( \phi \) is a YES instance. Then $\lindisc_\infty(\mat{A}') \leq 4/3$.\end{lemma}
\begin{proof}
The claim follows from the completeness argument in \cite[Proof of Lemma 5]{manurangsi21}. 
\end{proof}

\begin{lemma}[Soundness for $\Pi_2$-hardness reduction, \cite{manurangsi21}] \label{lem:soundness-Pi2-manurangsi}
Let \( \phi(V_A, V_E) \) be an $\AENAEEthreeSAT$ instance with $n' := \card{V_A}$ universally quantified variables and \( n := \card{V_A} + \card{V_E} \) variables, and let $\mat{A}' = \mat{A}'(\phi)$ be as defined in \cref{eqn:def-matAprime}. 
Suppose that for every $\vec{w}' := (\vec{w}, \vec{w}^*)$ for $\vec{w} := 1/2 \cdot \vec{1} \in \R^{3n}$ and $\vec{w}^* \in \set{1/3, 2/3}^n$ there exists a vector $\vec{x}' \in \bit^{3n + n'}$ such that $\norm{\mat{A}' (\vec{w}' - \vec{x}')}_{\infty} < 3/2$. Then $\phi$ is a YES instance of $\AENAEEthreeSAT$.
\end{lemma}
\begin{proof}
The claim follows from the soundness argument in \cite[Proof of Lemma 5]{manurangsi21}. 
\end{proof}

We now state a lemma that guarantees that if a vector $\vec{w}'$ of the form in \cref{lem:soundness-Pi2-manurangsi} and $\vec{x}' \in \Z^{3n + n'}$ satisfy \( \norm{\mat{A}'(\vec{w}' - \vec{x}')}_\infty < 3/2 \), then it must be the case that \( \vec{x}' \in \bit^{3n + n'} \). %

\begin{lemma} \label{lem:close-to-wprime-linfinity}
Let \( \phi(V_A, V_E) \) be an $\AENAEEthreeSAT$ instance with $n' := \card{V_A}$ universally quantified variables and \( n := \card{V_A} + \card{V_E} \) variables, and let $\mat{A}' = \mat{A}'(\phi)$ be as defined in \cref{eqn:def-matAprime}. Furthermore, let $\vec{w} := 1/2 \cdot \vec{1} \in \R^{3n}$, let $\vec{w}^* \in \set{1/3, 2/3}^{n'}$, let \( \vec{x} \in \Z^{3n} \), and let $\vec{x}^* \in \Z^{n'}$.
Let $\vec{w}' := (\vec{w}, \vec{w}^*) \in \R^{3n + n'}$ and let $\vec{x}' := (\vec{x}, \vec{x}^*) \in \Z^{3n + n'}$. 
Assume that $\norm{\mat{A}' (\vec{w}' - \vec{x}')}_{\infty} < 3/2$. Then $\vec{x}' \in \bit^{3n + n'}$.
\end{lemma}
\begin{proof}
We have that
\begin{equation} \label{eq:close-to-wprime-linfinity}
    \max_{j \in [n]}\norm{\mat{G}(\vec{w}_j - \vec{x}_j)}_\infty
    \leq \norm{\mat{A}(\vec{w} - \vec{x})}_\infty \leq \norm{\mat{A}'(\vec{w}' - \vec{x}')}_\infty  < 3/2
         \ \text{,}
\end{equation}
where the first inequality follows from the definition of $\mat{A}$ (\cref{eq:def-matA}), the second inequality follows from \cref{lem:pi2-structural-identity}, and the third inequality is an assumption.
Let $\vec{x}^1, \vec{x}^2, \vec{x}^3 \in \Z^n$ be such that $\vec{x} := (\vec{x}^2, \vec{x}^2, \vec{x}^3)$, and for $j \in [n]$, let $\vec{x}_j := (x^1_j, x^2_j, x^3_j) \in \Z^3$.
By \cref{eq:close-to-wprime-linfinity} we then have that for all \( j \in [n] \), \( \norm{\mat{G}(1/2 \cdot \vec{1} - \vec{x}_j)}_\infty < 3/2 \).
Therefore, by \cref{lem:gadget-property-soundness} we have have that for all \( j \in [n] \), \( \vec{x}_j \in \set{\vec{0}, \vec{1}} \). In particular, $\vec{x} \in \bit^{3n}$.

Again, by setting \( \vec{y}' := \vec{w}' - \vec{x}' \) in \cref{lem:pi2-structural-identity}, we have that \( (8/3) \cdot \max_{i \in [n']} \abs{w^*_i - x^*_i} < 3/2 \), or equivalently, that for all \( i \in [n'] \), \( \abs{w^*_i - x^*_i} < 9/16 < 2/3 \). It follows that if $w_i^* = 1/3$ then $x_i^* = 0$, and if $w_i^* = 2/3$ then $x_i^* = 1$. In particular, \( \vec{x}^* \in \bit^{n'} \), as needed.
\end{proof}

Finally, we obtain the following theorem that adapts \cite[Lemma 5]{manurangsi21} to the \( \BinGapCRP_\infty \) problem.
\begin{theorem} \label{thm:AENAE3SAT-to-BinGapCRP}
    For any constant $\eps > 0$, there is a polynomial-time mapping reduction from $\AENAEEthreeSAT$ to $(9/8 - \eps)$-\( \BinGapCRP_\infty \).
\end{theorem}
\begin{proof}
    Let \( \phi \) be an instance of the $\AENAEEthreeSAT$ problem. The reduction outputs the $\BinGapCRP$ instance with matrix \( \mat{A'} := \mat{A'}(\phi) \), as defined in \cref{eqn:def-matAprime}, and distance threshold $r := 4/3$.
    The reduction is efficient by inspection. 
     If \( \phi \) is a YES instance, then \( \lindisc_\infty(\mat{A}') \leq 4/3 \) by \cref{lem:completeness-Pi2-manurangsi}.
     
     On the other hand, if $\mu_{\infty}(\lat(\mat{A}')) < 3/2$, then by definition of the covering radius and  \cref{lem:close-to-wprime-linfinity} we have that for all $\vec{w}' := (\vec{w}, \vec{w}^*) \in \R^{3n + n'}$ with $\vec{w} := 1/2 \cdot \vec{1}^{3n}$ and $\vec{w}^* \in \set{1/3, 2/3}^{n'}$, there exists \( \vec{x}' \in \bit^{3n + n'} \) such that $\norm{\mat{A}' (\vec{w}' - \vec{x}')}_{\infty} < 3/2$. 
     Therefore, $\phi$ is a YES instance by \cref{lem:soundness-Pi2-manurangsi}. The result follows.
\end{proof}

The proof of \cref{thm:intro-pi2-hardness} now follows easily.

\begin{proof}[Proof of \cref{thm:intro-pi2-hardness}]
The result follows from the $\Pi_2$-hardness of $\AENAEEthreeSAT$ (shown in  \cite[Theorem 4.7]{docker-AENAE-E3-SAT}) and the reduction in \cref{thm:AENAE3SAT-to-BinGapCRP}.
\end{proof}

\bibliography{covering}
\bibliographystyle{alphaabbrvprelim}

\appendix
\section{NP-hardness of \texorpdfstring{\( (15/16 + \eps, 1) \)}{(15/16 + eps, 1)}-NAE-E3-SAT}
\label{sec:nae-e3-sat-hardness}

We start with the following theorem of H\r{a}stad, which shows that certain $4$-ary CSPs are hard to approximate.\footnote{In fact, \cite[Theorem 7.18]{hastad} is stated for CSPs over $\pmo$ rather than $\bit$.}

\begin{theorem}[{\cite[Theorem 7.18]{hastad}}] \label{thm:hastad-4ary-csp}
    Let $P : \bit^4 \to \bit$ be a predicate such that
    \[
    P^{-1}(0) \subseteq \set{(0,0,0,0), (0,0,1,1), (1,1,0,0), (1,1,1,1)} \ \text{,}
    \] 
    and let $\CSP_P$ denote the class of $\CSP$s where each constraint is $P$ applied to four distinct literals (Boolean variables $x_1, \ldots, x_n$ or their negations).
    Let $\alpha := \E_{\tau \sim \bit^4}[P(\tau) = 1]$.
    Then for every constant $\eps \in (0, 1 - \alpha]$, $(\alpha + \eps, 1)$-$\CSP_P$ is $\NP$-hard. 
\end{theorem}

The value $\alpha := \E_{\tau \sim \bit^4}[P(\tau) = 1]$ in \cref{thm:hastad-4ary-csp} is the expected value of a random assignment to a constraint involving $P$. Note that if $\phi(x_1, \ldots, x_n)$ is an instance of $\CSP_P$, then by linearity of expectation $\alpha$ is also the expected value of a random assignment to $\phi$ (i.e., $\alpha = \E_{\tau \sim \bit^n}[\val_{\tau}(\phi)]$).
So, the theorem says that it is $\NP$-hard to approximate $\CSP_P$ even slightly better than what is obtained from a random assignment.

We then obtain the following hardness result for $\NAEEfourSAT$ from \cref{thm:hastad-4ary-csp}.
\begin{corollary}
\label{cor:nae-e4-sat-hard}
    For any constant $\eps \in (0, 1/8)$, \( (7/8 + \eps, 1) \)-$\NAEEfourSAT$ is $\NP$-hard.
\end{corollary}

\begin{proof}
Set $P : \bit^4 \to \bit$ to be such that \( P^{-1}(0) := \set{(0,0,0,0), (1,1,1,1)} \), and note that $\NAEEfourSAT$ formulas are exactly the class of CSPs where each constraint is $P$ applied to four literals.
The result then follows by \cref{thm:hastad-4ary-csp} and the observation that $\E_{\tau \sim \bit^4}[P(\tau) = 1] = 7/8$ for such~$P$.
\end{proof}

We next give a reduction from approximate $\NAEEfourSAT$ to approximate $\NAEEthreeSAT$.

\begin{theorem}
    \label{thm:nae-e4-sat-to-nae-e3-sat}
    There is a polynomial-time reduction from NAE-E4-SAT instances \( \phi \) to NAE-E3-SAT instances \( \phi' \) such that \( \mathrm{val}(\phi') = \frac{1}{2} (1 + \mathrm{val}(\phi)) \).
\end{theorem}
\begin{proof}
    Let $P$ denote the ``not all equal'' predicate on some number of Boolean inputs.
    Let $\phi(x_1, \ldots, x_n)$ be the input $\NAEEfourSAT$ instance, and assume that $\phi$ has \( m \) constraints $C_1, \ldots, C_m$, each of which is equal to $P(\ell_1, \ell_2, \ell_3, \ell_4)$ for four literals $\ell_1, \ell_2, \ell_3, \ell_4$.
    For $i = 1, \ldots, m$, the reduction introduces the new variable $s_i$, and applies the map
    \[
    C_i := P(\ell_1, \ell_2, \ell_3, \ell_4) \mapsto \set{C_i := P(\ell_1, \ell_2, s_i), C_i'' := P(\ell_3, \ell_4, \lnot s_i)} \ \text{.}
    \]
  It then outputs the $\NAEEthreeSAT$ formula $\phi'$ consisting of the set of $2m$ constraints $\cup_{i=1}^m \set{C_i', C_i''}$.

    It is clear that the reduction runs in polynomial time, and so it remains to show that $\val(\phi') = \frac{1}{2} (1 + \val(\phi))$. 
    We first show that $\val(\phi') \geq \frac{1}{2} (1 + \val(\phi))$.
    Let $\tau \in \bit^n$ be an assignment such that $\val_{\tau}(\phi) = \val(\phi)$.
    We construct an assignment $\tau' \in \bit^{n + s}$ by assigning the variables $x_1, \ldots, x_n$ as in $\tau$, and assigning the variables $s_1, \ldots, s_m$ so that:
    \begin{enumerate}
    \item \label{item:Ci_unsat} If $C_i$ is not satisfied, one of $C_i', C_i''$ is satisfied.
    \item \label{item:Ci_sat} If $C_i$ is satisfied, both of $C_i', C_i''$ are satisfied.
    \end{enumerate}

    \cref{item:Ci_unsat} is possible to achieve, e.g., by assigning $s_i := \lnot \ell_i(\tau)$, where we abuse notation slightly and write $\ell_i(\tau)$ to denote the evaluation of $\ell_i$ on $\tau$.
    \cref{item:Ci_sat} is possible to achieve by considering the following two cases.
    First, if $\ell_3(\tau) \neq \ell_4(\tau)$, then $C_i''$ is satisfied regardless of $s_i$, and so we assign $s_i := \lnot \ell_1(\tau)$ in order to satisfy $C_i'$.
    Second, if $\ell_3(\tau) = \ell_4(\tau)$, then we assign $s_i := \ell_3(\tau)$. We claim that both $C_i'$ and $C_i''$ are satisfied. Indeed, since $\tau$ satisfied $C_i$ and since $s_i =\ell_3(\tau)$, we must have that either $\ell_1(\tau) \neq \ell_3(\tau) = s_i$ or $\ell_2(\tau) \neq \ell_3(\tau) = s_i$. It follows that $C_i'$ is satisfied.
    Additionally, $C_i''$ is satisfied since $\lnot s_i \neq \ell_3(\tau)$.
    We then have that $\val(\phi') \geq \val_{\tau'}(\phi') = (2 \val(\phi) \cdot m + (1 - \val(\phi)) \cdot m)/(2m) = \frac{1}{2} (1 + \val(\phi))$, as needed.

    We now show that $\val(\phi') \leq \frac{1}{2} (1 + \val(\phi))$.
    Let $\tau' \in \bit^{n + s}$ be an assignment such that $\val_{\tau'}(\phi') = \val(\phi')$. 
    Let $A \in \set{0, 1, 2}$ be the random variable defined by sampling uniformly random $i \sim [m]$ and reporting how many of the constraints in the pair $C_i', C_i''$ are satisfied by $\tau'$. Note that $\E[A] = 2 \val_{\tau'}(\phi') = 2 \val(\phi')$.
    We claim that $\val(\phi) \geq \val_{\tau}(\phi) \geq \Pr[A = 2]$, where $\tau \in \bit^n$ is the restriction of $\tau'$ to the variables $x_1, \ldots, x_n$.
    Indeed, if $\tau'$ satisfies a pair of constraints $C_i' = P(\ell_1, \ell_2, s_i), C_i'' = P(\ell_3, \ell_4, \lnot s_i)$ in $\phi'$ for some $i \in [m]$ then $\tau$ satisfies $C_i$ in $\phi$ since $C_i', C_i''$ are both satisfied only if $\ell_1(\tau'), \ell_2(\tau'), \ell_3(\tau'), \ell_4(\tau')$ are not all equal.

    By applying Markov's inequality to the nonnegative random variable $2 - A$, we have that
    \begin{equation} \label{eq:markov-nae3sat}
   \Pr[A = 1] \leq \Pr[A \leq 1] = \Pr[2 - A \geq 1] \leq \E[2 - A] = 2(1 -  \val(\phi')) \ \text{.}
    \end{equation}
    Therefore,
    \begin{align*}
    \val(\phi') &= \frac{1}{2} \E[A] \\
                &= \frac{1}{2} (2 \Pr[A = 2] + \Pr[A = 1]) \\
                &\leq \Pr[A = 2] + 1 - \val(\phi') \\
                &\leq \val(\phi) + 1 - \val(\phi')
                \ \text{,}
    \end{align*} 
    where the first inequality uses \cref{eq:markov-nae3sat} and the second inequality uses the claim that $\val(\phi) \geq \Pr[A = 2]$ above.
    The inequality $\val(\phi') \leq \frac{1}{2} (1 + \val(\phi))$ follows by rearranging.

\end{proof}

We are now able to conclude our main theorem about the $\NP$-hardness of $(15/16 + \eps, 1)$-$\NAEEthreeSAT$ as a corollary.
\begin{proof}[Proof of \cref{thm:nae-e3-sat-hard}]
    Combine \cref{cor:nae-e4-sat-hard,thm:nae-e4-sat-to-nae-e3-sat}.
\end{proof}

\end{document}